\documentclass[preprint,12pt]{elsarticle}

\usepackage{epsfig}

\usepackage{amssymb}
\usepackage[cmex10]{amsmath}
\usepackage{amsthm}
\usepackage{textcomp}

\makeatletter
\newcommand{\newreptheorem}[2]{%
\newenvironment{rep#1}[1]{%
 \def\rep@title{#2 \ref{##1}}%
 \begin{rep@theorem}}%
 {\end{rep@theorem}}}
\makeatother

\newtheorem{thm}{Theorem} 
\newtheorem{lem}{Lemma}
\newtheorem{define}{Definition}

\newtheorem{conj}{Conjecture}

\newreptheorem{thm}{Theorem} 
\newreptheorem{lem}{Lemma}

\usepackage{algorithm}
\usepackage{algorithmic} 
\makeatletter
\renewcommand{\ALG@name}{Procedure}
\renewcommand{\listalgorithmname}{List of \ALG@name s}
\makeatother

\usepackage[export]{adjustbox}

\journal{ }

\newcommand{\morphosys}{{\sc{MorphoSys}}}
\newcommand{\MORPHOSYS}{{\normalsize\bf{\large\bf{M}}ORPHO{\large\bf{S}}YS}}
\begin{document}

\begin{frontmatter}



\title{MORPHOSYS: Efficient Colocation of QoS-Constrained Workloads in the Cloud\tnoteref{t1}}

\author[vi]{Vatche Ishakian}
\ead{vishakian@bentley.edu}
\author[fac]{Azer Bestavros}
\ead{Best@bu.edu}
\author[fac]{Assaf Kfoury}
\ead{kfoury@bu.edu}


\address[vi]{Bentley University, Computer Information Systems, Waltham, MA}
\address[fac]{Boston University, Computer Science Department, Boston, MA}
\tnotetext[t1]{\scriptsize Initial version of this work was published in \cite{ishakian2012morphosys}}

\begin{abstract}
 In hosting environments such as IaaS clouds, desirable application
  performance is usually guaranteed through the use of Service Level
  Agreements (SLAs), which specify minimal fractions of resource
  capacities that must be allocated for  use for proper 
  operation. Arbitrary colocation of applications with different SLAs
  on a single host may result in inefficient utilization of the host's
  resources. In this paper, we propose that periodic resource
  allocation and consumption models 
  be used for a more granular expression of
  SLAs. Our proposed SLA model has the salient feature that it exposes
  flexibilities that enable the IaaS provider to safely
  transform SLAs from one form to another for the purpose of achieving
  more efficient colocation.  Towards that goal, we present
  \morphosys: a framework for a service that allows the manipulation
  of SLAs to enable efficient colocation of workloads. We present results from extensive trace-driven
  simulations of colocated Video-on-Demand servers in a cloud
  setting. The results show that potentially-significant reduction
  in wasted resources (by as much as 60\%) are possible using
  \morphosys.
\end{abstract}

\end{frontmatter}


\section{Introduction}
\noindent {\bf Motivation:} {\em Cloud computing} in general and
  Infrastructure as a Service (IaaS) in particular have emerged as
  compelling paradigms for the deployment of distributed applications
  and services on the Internet due in large part to the maturity and wide  adoption of virtualization. 

From the perspective of an IaaS customer, this paradigm shift presents
  both an opportunity and a risk. On the one hand, deploying
  applications in the cloud is attractive because it enables
  efficiency through elastic scaling. On the
  other hand, deploying applications in the cloud implies
  relinquishing QoS monitoring and control to the cloud. Mitigating
  that risk requires the establishment of a ``contract'' -- a Service
  Level Agreement (SLA) -- between the provider and the customer,
  which spells out minimal resource allocations that the customer
  believes would satisfy desirable QoS constraints, while also being
  verifiable through measurement or auditing of allocated
  resources. Indeed, providing trustworthy accountability and
  auditing features have been cited as key attributes that would
  increase cloud adoption \cite{haeberlen2009case,sripanidkulchai2009clouds}. 

From the perspective of an IaaS provider, the cloud value proposition is
  highly dependent on efficient resource management
  \cite{Wood2009Memory,podzimek2015analyzing}, reduced operational costs \cite{Cardosa2009Shares} and on
  improved scalability \cite{Meng2010Improving}. Such efficiencies
  need to be achieved while satisfying the aforementioned
  contractually-binding customer SLAs. This necessitates that SLAs be
  spelled out in such a way so as to expose potential flexibilities
  that enable efficient mapping of physical resources to virtualized
  instances.

Given the wide range of applications currently supported in an IaaS
  setting, 
  it would be impractical for an IaaS provider to support
  special-purpose SLAs that are tailor-made for each such application
  and service, and which can be efficiently audited. Rather, a more
  practical approach calls for the development of a common language
  for expressing SLAs -- a language that would cater well to the
  widely different types of applications that are likely to be
  colocated on an IaaS infrastructure.

Currently, the de-facto language for expressing SLAs mirrors how
  virtual machines are provisioned -- namely through the specification
  of resource capacities to be allocated on average, over fairly long
  time scales.  While appropriate for many applications, such coarse
  SLAs do not cater well to the needs of applications that require
  resource allocations at a more granular scale.  
To elaborate, recent studies have documented the often unacceptable or degradable
  performance of a number of application classes in a cloud
  setting. Examples include latency-sensitive, interactive, web applications \cite{li2014impact}, image
  acquisition applications, IP telephony and streaming applications \cite{barker2010empirical,NetflixProblem}. 
  A culprit for the degraded performance is lack of any guarantee associated with the time-scale of resource allocation in a   virtualized environment \cite{baset2012cloud}. Indeed, to provide QoS features, which are becoming the differentiating elements between cloud computing environments, 
  there is the need for finer-grain SLA specifications that enable applications to spell out their resource   needs over arbitrary time scales, as well as any tolerable   deviations thereof (flexibilities).

Recognizing this need, in this paper we propose an expressive periodic
  resource allocation model for the specification of SLAs -- a model
  that on the one hand provides customers with a larger degree of
  control over the granularity of resource allocation, and on the
  other hand enables providers to leverage flexibilities in customers'
  SLAs for the efficient utilization of their infrastructures. Our SLA
  model is equally expressive for traditional cloud application as
  well as for the aforementioned QoS-constrained applications; it enables
  providers to cater to a wider customer base while providing them
  with the requisite measurement and auditing
  capabilities. 

 \noindent {\bf Scope and Contributions:} Given a set of applications
  (workloads), each of which specified by minimal resource utilization
  requirements (SLAs), the problem we aim to address is that of
  colocating or mapping  these workloads efficiently to physical resources. To
  achieve efficient mapping, we need to provide workloads with the
  ability to express potential flexibilities in satisfying their
  SLAs. 
Recognizing that there could
  be the case where there are multiple, yet functionally equivalent
  ways to express the resource requirements of a QoS-constrained
  workload. Towards that end, we propose a specific model for SLAs
  that makes it possible for providers to rewrite such SLAs as long as
  such rewriting is safe. By safety, we indicate that we can substitute the original SLA by the rewritten SLA without violating the original SLA, and that the resources allocations that satisfy the rewritten SLA would also provably satisfy the original SLA. The ability to make such safe SLA
  transformations enables providers to consider a wider range of
  colocation possibilities, and hence achieve better economies of
  scale. In that regard, we present \morphosys:\footnote{\morphosys\
  can be seen as catalyzing the ``morphosis'' of a set of SLAs --
  namely, morphing SLAs to enable more efficient colocation.} the
  blueprints of a colocation service that demonstrates the premise of
  our proposed framework. Results from extensive trace-driven simulations
  of colocated Video-on-Demand (VOD) servers in a cloud setting
  show that potentially-significant reduction in wasted resources (by
  as much as 60\%) are possible using
  \morphosys.

\noindent {\bf Paper Overview:} The remainder of this paper is
  organized as follows. In Section \ref{sec:background} we present
  some background and illustrative examples that motivate the need for an expressive safe SLA model. In Section \ref{sec:basics} we introduce our basic type-theoretic model for periodic resource supply and demand with necessary notation, basic definitions, and a series of safe transformations as exemplars of our notion of safe SLA rewrite rules.  In Section \ref{sec:SLAModel}, we extend our SLA model for
  QoS-constrained resource supply and demand. In section \ref{sec:DynamicService}, we present the
  basic elements of our \morphosys\ framework. In Section
  \ref{sec:experimental}, we present experimental results that
  demonstrate the promise from using \morphosys\ to manage colocated
  streaming servers. We review related work in Section
    \ref{sec:relatedwork}, and provide a conclusion in Section \ref{sec:conclusion}.

\section{background \lowercase{and} illustration} \label{sec:background}

Recall that an important consideration for efficient colocation is the
  ability of a provider to evaluate whether a given set of customers
  can be safely colocated. To do so, a provider must be able to decide
  whether the capacity of a given set of resources ({\em e.g.,} a
  host) can satisfy the aggregate needs of a set of customers (namely,
  the composition of the SLAs for that set of customers).
Given our
  adopted periodic model for SLA specification, it follows that
  evaluating the feasibility of colocating a set of customer workloads
  on a given host can be viewed as a ``schedulability'' problem: given
  the capacity of a host, are a set of periodic real-time tasks
  schedulable? 

Different models and schedulability analysis techniques have been
  proposed in the vast real-time scheduling theory, including Earliest
  Deadline First, Rate Monotonic Analysis (RMA)
  \cite{Liu1973Scheduling}, 
among others \cite{R:Davis:2009d}. While similar in terms of
  their high-level periodic real-time task models, these approaches
  differ in terms of the trade-offs they expose vis-a-vis the
  complexity of the schedulability analysis, the nature of the
  underlying resource manager/scheduler, and the overall achievable
  system utilization. Without loss of generality, we assume that RMA
  \cite{Liu1973Scheduling} is the technique of choice when evaluating
  whether it is possible to co-locate a set of periodic workloads on a
  fixed-capacity resource.\footnote{While the analysis and transformations we provide in this paper are based
  on RMA, we emphasize that our framework and many of our results naturally
  extend to other types of schedulers.}

Liu and Layland \cite{Liu1973Scheduling} provided the following
  classical result for the schedulability condition of $n$
  tasks (SLAs), each of which requiring the use of a resource for $C_i$
  out of every $T_i$ units of time, under RMA:

	$$\mbox{\fontsize{9}{9}\selectfont $\displaystyle U = \sum_{i=1}^{n} \frac{C_i}{T_i} \leq n(\sqrt[n]{2} - 1)$}$$

  Follow-up work, by Lehoczky {\em et al} \cite{Lehoczky1989Rate} 
  showed that by grouping tasks in $k$ clusters such that the periods of tasks in each cluster are multiples of each other
({\em i.e.}, harmonic), a tighter schedulability condition is possible -- namely:

$$\mbox{\fontsize{9}{9}\selectfont $\displaystyle U = \sum_{i=1}^{n} \frac{C_i}{T_i} \leq k(\sqrt[k]{2} - 1)$}$$

As motivated above, there may be multiple yet functionally-equivalent
  ways to satisfy a given SLA. This flexibility could be leveraged by
  a provider for efficient colocation. In particular, given a set of
  periodic tasks (SLAs), it might be possible to obtain clusters of
  tasks with harmonic periods by manipulating the period $C_i$ or the periodic allocation $T_i$ of some
  of the tasks (SLAs) in the set. For such a transformation to be
  possible, we must establish that it is {\em safe} to do so.
\begin{figure}[h]
  \centering
  \includegraphics[width=0.65\textwidth]{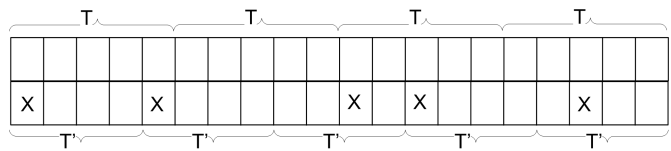}
  \caption{\label{expfigure} Illustration: Reducing the allocation period may result in missed deadlines.}
\end{figure}
To illustrate why a transformation may not be safe (even if it results in an increase in the fraction of the resource alloted to the task), consider a workload that requires $C=1$ time units of the resource every period $T=5$ time units. While reducing the allocation period for this task from $T=5$ to $T'=4$ would result in that task being alloted the resource for a larger fraction of time (25\% as opposed to 20\%), as shown in Figure \ref{expfigure}, it is possible for that task to miss its original deadlines. The figure shows a sequence of allocation intervals of size $T$ and $T'$, where a box represents a single allocation period. The upper row shows the periodic boundaries as originally specified $(T=5)$, whereas the lower row shows a periodic allocation with $(T'=4)$, with ``X'' marking the times when the resource is allocated. One may observe that an entire period of size $T$ is missing its required allocation.

In the above example, the fact that the transformation we considered resulted (or may result) in missed deadlines does not mean that it cannot be used. In particular, if the SLA associated with the workload in question allows for some percentage of deadline misses, then if one is able to bound the deadline misses resulting from the transformation -- and consequently show that the SLA is not violated -- then the transformation is indeed safe. Thus the need for a concise workload model along with a formalism for studying such safe transformations.

\section{SLA Model: Basics} \label{sec:basics}
\noindent
As we established earlier, SLAs can be seen as encapsulators of
  the resources supplied by hosts (producers) and demanded by
  tasks (consumers). While this concept is generic enough for a
  wide variety of resources, in this section, we provide a specific model for SLAs --
  namely, one that supports periodic, real-time resource supply
  and demand.\footnote{Legal implications that specify how penalties are associated with SLA violations are considered to be out of scope of our work.} 
We also provide the basic type-theoretic-inspired definitions that  allow us to establish subtyping relationships between SLAs.

Although our SLA formulation and subtyping relationships is abstract enough to reflect an aggregated set of resources -- in a public or private IaaS setting -- such as a single rack in the datacenter, throughout this paper, we assume that our SLAs reflect the resources provided by a single physical host.


\subsection{Periodic Supply/Demand SLA Types}
This section presents the formal definition of SLA types for resources supplied and demanded for a specific allocation interval. It also denotes the maximum number of missed allocations over multiple intervals.

\begin{define}
A Service Level Agreement (SLA) type $\tau$ is defined as a
  quadruple of natural numbers $(C,T,D,W)$, $C \leq T$, $D
  \leq W$, and $W \geq 1$, where $C$ denotes the resource
  capacity supplied or demanded in each allocation interval $T$,
  and $D$ is the maximum number of times such an allocation is
  not possible to honor in a window consisting of $W$
  allocation intervals.
\end{define}

As is common in the real-time literature, the above definition assumes that the periodic capacity could be allocated as early as the beginning of any interval (or period) and must be completely produced/consumed by the end of that same interval ({\em i.e.}, allocation deadline is T units of time from the beginning of the period).


The concept of SLA types is general enough to capture the
  various entities in a resource allocation. The
  following are illustrative examples.

An SLA of type $(1,1,0,1)$ could be used to characterize a
  uniform, unit-capacity supply provided by a physical host. An
  SLA of type $(1,n,0,1)$, $n > 1$ could be used to characterize
  the fractional supply provided under a General Processor
  Sharing (GPS) model to $n$ processes. 
 In the above examples,
  the SLA type does not admit missed allocations (by virtue of
  setting $D=0$).

An SLA of type $(1,30,0,1)$ could be used to represent a task that
needs a unit capacity $C=1$ over an allocation period $T=30$
  and cannot tolerate any missed allocations. An SLA of type $(1,30,2,5)$ is similar in its
  periodic demand profile except that it is able to tolerate
  missed allocations as long as there
  are no more than $D=2$ such misses in any window of $W=5$
  consecutive allocation periods.

\subsection{Satisfaction and Subtyping of SLAs}
We begin by providing basic definitions of what it means to satisfy a schedule for SLAs of the type $(C,T,0,1)$ ({\em i.e.}, those that do not admit missed allocations), which we denote using the shorthand $(C,T)$. Next, we generalize these definitions for general SLA types of the form $(C,T,D,W)$.

\begin{define}
  A schedule $\alpha$ is a function from $\mathbb{N}$ to $\{0,1\}$ as $\alpha:\mathbb{N} \rightarrow \{0,1\}$
 \end{define}

 A schedule $\alpha$ {\em satisfies} (denoted by $\vDash$) an SLA type $(C,T)$ if the resource is allocated for $C$ units of time in non-overlapping intervals of length $T$.

\begin{define}
  $\alpha \vDash (C,T)$ iff for every $Q \geq 0$ where $Q \in \mathbb{N}$ and every $m = Q \times T$ we have $\alpha(m) + \cdots + \alpha(m + T -1) \geq C$
  \end{define}

There are multiple ways for an allocation to satisfy a particular schedule. Thus, we define the set $\textlbrackdbl (C,T) \textrbrackdbl$ to consists of all schedules that satisfy $(C,T)$.

\begin{define}
  $\textlbrackdbl (C,T) \textrbrackdbl = \{ \alpha:\mathbb{N} \rightarrow  \{0,1\}\;|\;\alpha \vDash (C,T) \}$
\end{define}

We are now ready to introduce SLA {\em subtyping} relationships (denoted by $\lhd$) as follows:
\begin{define}
  $(C,T) \lhd (C',T')$ iff $\textlbrackdbl (C,T) \textrbrackdbl  \subseteq \textlbrackdbl (C',T') \textrbrackdbl$
\end{define}

We generalize the above definitions for SLAs of the type $(C,T,D,W)$, which allow missed allocations. The definitions are conceptually similar to SLAs of type (C,T), though considerations of missed allocations require  more elaborate notations.

To calculate the number of missed allocations over $W$ intervals, we need to identify whether a single interval of size $T$ is satisfied. Formally,
\begin{define}
  $\mathbf{A}_{\alpha, C,T}$ is a function from $\mathbb{N}$ to $\{0,1\}$.  $\mathbf{A}_{\alpha, C,T}:\mathbb{N} \rightarrow \{0,1\}$ such that:
$$
  \mathbf{A}_{\alpha, C,T}(m) = \left\{
  \begin{array}{rl}
    0 &\mbox{ if $\alpha(m) + \cdots + \alpha(m + T -1) < C$} \\
    1 &\mbox{ if $\alpha(m) + \cdots + \alpha(m + T -1) \geq C$}
  \end{array} \right.
$$
\end{define}

A schedule $\alpha$ {\em satisfies} an SLA type $(C,T,D,W)$ if the resource is allocated for $C$ units of time in at least $W-D$ out of every $W$ intervals of length $W \times T$.

\begin{define}
    $\alpha \vDash (C,T,D,W)$ iff for every $Q \geq 0$ where $Q \in \mathbb{N}$ and every $m = Q \times W \times T$
      \begin{align*}
        &\mathbf{A}_{\alpha, C,T}(m) + \mathbf{A}_{\alpha, C,T}(m+T) + \cdots +
      \mathbf{A}_{\alpha, C,T}(m+ (W-1) \times T) \geq W-D
      \end{align*}
\end{define}

The set $\textlbrackdbl (C,T,D,W) \textrbrackdbl$ is defined as comprising all schedules that satisfy $(C,T,D,W)$. Formally:

\begin{define}
    $\textlbrackdbl (C,T,D,W) \textrbrackdbl = \{\alpha:\mathbb{N} \rightarrow \{0,1\}\;|\;\alpha \vDash(C,T,D,W) \}$.
\end{define}

A generalized notion of {\em subtyping} for SLAs of the form $(C,T,D,W)$ is presented as follows:
\begin{define}
  $(C,T,D,W) \lhd (C',T',D',W')$ iff $\textlbrackdbl (C,T,D,W) \textrbrackdbl  \subseteq \textlbrackdbl (C',T',D',W') \textrbrackdbl$.
\end{define}

\subsection{SLA Subtyping and Transformations} \label{Characterize}
\noindent
In this section, we present a set of SLA transformations that exemplify the range of scheduling results that could be ``coded into'' the \morphosys\ framework. Each one of the transformations presented in this section is cast within a subtyping theorem. Intuitively, establishing a subtyping relationship between two SLAs implies that we can {\em safely} substitute one for the other 
In general, transformed SLAs will require more of the underlying physical resources than the original SLA. Nevertheless, such a transformation may be advantageous to the IaaS provider as it may result in a more efficient colocation of customer workloads -- {\em e.g.}, by making such workloads harmonic and hence subject to looser schedulability bounds \cite{Lehoczky1989Rate}.

Given an SLA of type $(C',T')$, it is possible to safely transform it into another SLA of type $(C,T)$, where $T$ is larger or smaller than $T'$. 

\begin{thm}\footnote{Most of the proofs in this paper were verified using the theorem prover of \cite{Ishakian2011Formal}.}\label{thm2}
  $(C,T) \lhd (C',T')$ iff one of the following conditions holds:
    \begin{enumerate}
        \item $T \leq T'/2$ and $C \geq C'/(K-1)$ where $K = \lfloor T'/T \rfloor$.
        \item $T > T'$ and $C \geq T - (T' - C')/2$.
        \item $T'/2 < T \leq T'$ and $T - (T' - C')/3 \leq C$.
    \end{enumerate}
\end{thm}
\begin{proof}
\noindent{\em Condition 1:} [If] $T \leq T'/2$ implies that $K \geq
  2$. According to Lemma \ref{k1overlap} (see appendix), an
  interval of length $T$ must overlap with at least $K-1$ fixed
  intervals of length $T'$. These $K-1$ intervals provide an
  allocation of $(K-1) \times C'$, enough for interval $T$. Thus,
  $(K-1)  \times  C' \geq C $ and $C' \geq C/(K-1)$.

\noindent{\em Condition 2:} [If] Consider any interval $I'$ of length $T'$. Since $T > T'$,
  either $I'$ will be completely overlapped by an
  interval $I$ of length $T$, or it will be overlapped by two
  intervals $I_1$, $I_2$ of length $T$ (as shown in Figure \ref{weakcase2}).
  For any interval $I$ of length $T$, denote the
  left and right boundaries of $I$ using $l(I)$ and
  $r(I)$, respectively.

 \begin{figure*}[htp]
  \begin{minipage}{0.47\linewidth}
  \includegraphics[width=1\textwidth]{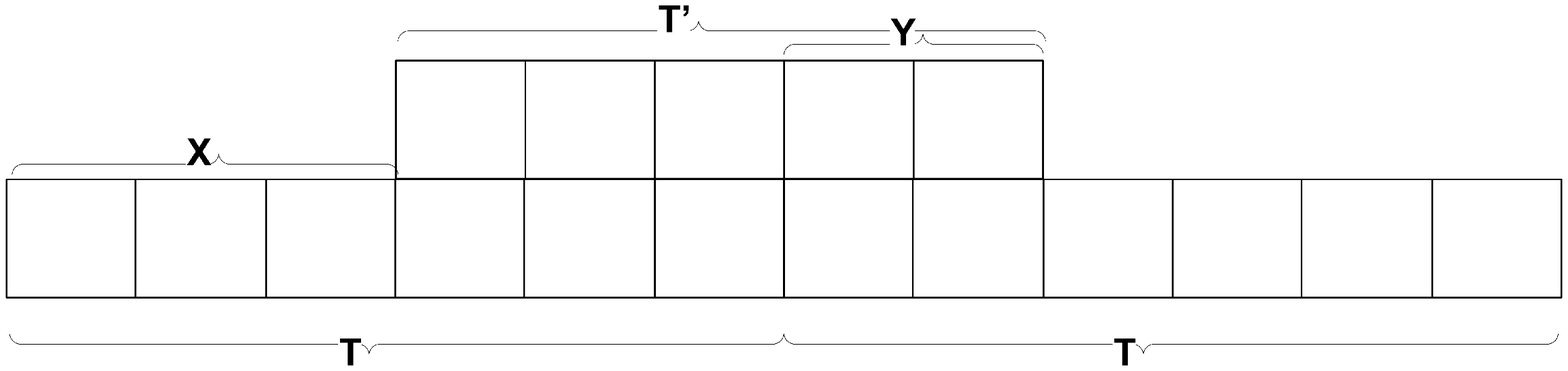}\\ \vspace{-0.25in}
  \caption{\label{weakcase2}$T'$ is overlapped by two intervals of size $T$}
  \end{minipage}
  \begin{minipage}{0.06\linewidth}
   $\quad$
  \end{minipage}
    \begin{minipage}{0.47\linewidth}
  \includegraphics[width=1\textwidth]{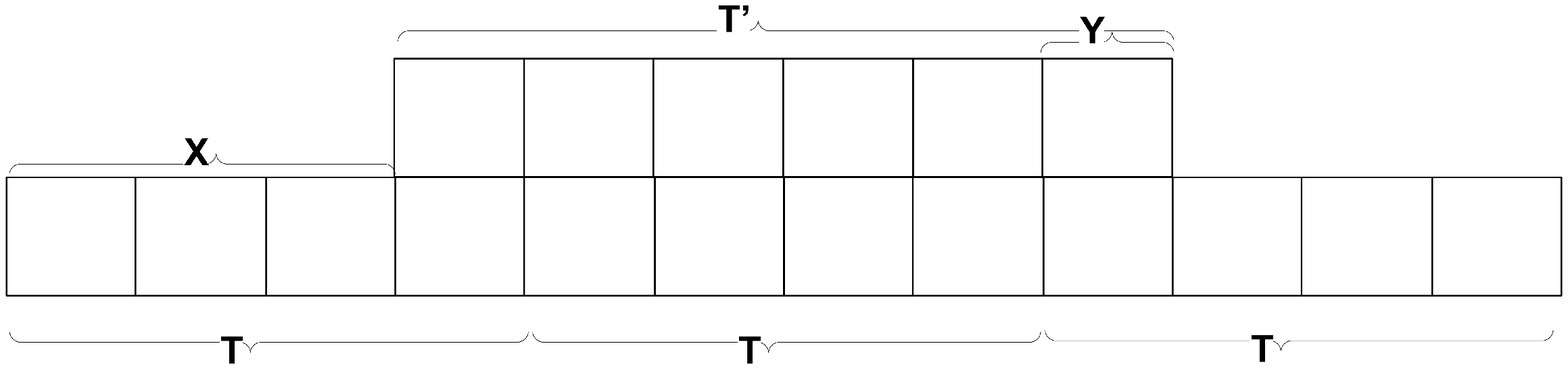}\\ \vspace{-0.25in}
  \caption{\label{overlapsthreewinds}$T'$ is overlapped by three intervals of size $T$}
    \end{minipage}
\end{figure*}

  Let $x$ be the offset of $l(I')$ from $l(I_1)$ and $y$ be the offset of
  $r(I')$ from $l(I_2)$. We observe that $(T' - y) + x =
  T$, leading to $C \geq 1/2(x +
  (T-y) + C')$ and $C \geq T - (T' - C')/2$ as a
  sufficient condition.

\noindent{\em Condition 3:} [If] Consider any interval $I'$ of
  length $T'$. Since $T'/2 < T \leq T'$, $I'$ will overlap with
  either two or three intervals of length $T$. The case in which
  $I'$ overlaps two intervals of length $T$ follows from
  Condition 2, resulting in $C \geq (2T' - T + C')/2$. The case
  in which $I'$ overlaps three intervals $I_1$, $I_2$, and $I_3$
  of length $T$ is shown in Figure \ref{overlapsthreewinds}.

  Let $x$ be the offset of $l(I')$ from $l(I_1)$ and $y$ be the
  offset of $r(I')$ from $l(I_3)$. We observe that $(T' - y) + x
  =2T$. Thus, a sufficient condition is $C \geq 1/3(x + (T-y) + C)$ and
  $C \geq T - (T' - C')/3$.
  Thus a bound for both cases is the maximum of the two
  bounds, namely $C \geq T - (T' - C')/3$.

\noindent{\em Condition 1:} [Only If] Suppose $\alpha \vDash
  (C,T)$. This implies that for every $Q \geq 0$ and every $m =
  Q \times T$: $\alpha(m) + \cdots + \alpha(m + T -1) \geq C$. Given
  that $K  \times  T \leq T'$ and $T \leq T'/2$, we have $K \geq 2$ and
  $(K-1)  \times  T < T'$. Given that $(K-1)  \times  C \geq C'$, we get:
  $\alpha(m) + \cdots + \alpha(m + ((K  -1)  \times  T) - 1) \geq K - 1  \times  C \geq C' $
  Therefore $(C,T) \lhd (C',T')$.

 \noindent{\em Condition 2:} [Only If] Suppose $\alpha \vDash
  (C,T)$. This implies that for every $Q \geq 0$ and every $m =
  Q \times T$: $\alpha(m) + \cdots + \alpha(m + T -1) \geq C$. Consider $m' = Q'  \times  T'$ for $Q' \geq 0$. In the worst case, it will will be overlapped by two intervals of length $m$. Thus we have $T' = 2T - x - x'$ where $x$ and $x'$ are the left overs from both intervals. Therefore $\alpha(m) + \cdots + \alpha(m + T - 1)  + \alpha(m) + \cdots + \alpha(m + 2T - 1 -x - x') \geq 2  \times  C -x -x' \geq  2C + T' - 2T  \geq C'$ Therefore $(C,T) \lhd (C',T')$.

 \noindent{\em Condition 3:} [Only If] Suppose $\alpha \vDash
  (C,T)$. This implies that for every $Q \geq 0$ and every $m =
  Q \times  T$: $\alpha(m) + \cdots + \alpha(m + T -1) \geq C$. Consider $m' = Q'  \times  T'$ for $Q' \geq 0$. In the worst case, it will will be overlapped by three intervals of length $m$. Thus we have $T' = 3T - x - x'$ where $x$ and $x'$ are the left overs from both intervals. Therefore $\alpha(m) + \cdots + \alpha(m + T - 1)  + \alpha(m) + \cdots + \alpha(m + 3T - 1 -x - x') \geq 3  \times  C -x -x' \geq  2C + T' - 3T  \geq C'$ Therefore $(C,T) \lhd (C',T')$.
  \end{proof}

We extend results from Theorem \ref{thm2} to provide safe transformations of SLAs of the type $(C,T,D,W)$. Theorem \ref{thm15} outline the necessary conditions to safely transform an SLA of type $(C',T',D',W')$ into $(C,T,D,W)$.

\begin{thm} \label{thm15} $(C,T,D,W) \lhd (C',T',D',W')$ if
  one of the following conditions holds:
\begin{enumerate}
  \item ($T \leq T'/2$ and $C \geq  C'/(K-1)$) and $D
    \leq D'/2$ and $W \geq D  \times  W'/D'  \times  (K + 1)$ where $K =
    \lfloor T'/T \rfloor$.
  \item ($T > T'$ and $C \geq T - (T' - C')/2$) and $D \leq
    D'/2K$ and  $W \geq D  \times
    W'/D'  \times  (K + 1)$ where $K = \lfloor T'/T \rfloor$.
  \item ($T'/2 < T \leq T'$ and $T - (T' - C')/3 \leq C$)
    and $D \leq D'/2$ and  $W \geq 2  \times  D  \times  W'/D'$.
    \end{enumerate}
\end{thm}

\begin{proof} We use the fact from Lemma \ref{Necessary}
(see appendix) that $D/W \leq D'/W'$ is a necessary condition.

\noindent{\em Condition 1:} The proof for the bracketed part
  of the conjuntion is identical to that under Condition 1 of
  Theorem \ref{thm2}. For the remaining part, we note that since
  missed deadlines might be stacked at the end of one window and
  at the beginning of the next contributing to a window of size
  W, it follows that $D \leq D'/2$.  Also, since $K+1$
  consecutive intervals of length $T$ will span one interval of
  length $T'$, it follows that every missed interval of length
  $T$ out of $K+1$ intervals will result in missing an interval
  of length $T'$. Thus, $W \geq (K+1)  \times  (D  \times  W/D')$ must hold.

\noindent{\em Condition 2:} The proof for the bracketed part
  of the conjuntion is identical to that under Condition 2 of
  Theorem \ref{thm2}. For the remaining part, in the worst case,
  missing an interval of length $T$ results in missing up to $(K
  + 1)  \times  T'$ intervals, where $K = \lfloor T'/T \rfloor$. Thus
  $D \leq D'/(K + 1)$ must hold as well as $W \geq (K+1)  \times  (D  \times
  W/D')$. However, since missed deadlines might be stacked at
  the end of one window and at the beginning of the next
  contributing to a window of size W, it follows that $D \leq
  D'/2(K+1)$ must hold.

\noindent{\em Condition 3:} The proof for the bracketed part
  of the conjuntion is identical to that under Condition 3 of
  Theorem \ref{thm2}. For the remaining part, the proof is
  similar to that in Condition 2 by taking $K = \lfloor T'/T
  \rfloor$ and consequently $K = 1$. Thus, $W \geq (K+1)  \times  (D  \times
  W/D')$ must hold.
\end{proof}

Having characterized some basic notions of subtyping, we present additional transformations that allow for safe rewriting of such types.

\begin{thm} \label{lemKCKT}
Let $\tau = (KC, KT)$ be an SLA type for some $K \geq 1$ and $\tau' = (C, T)$ be a host-provided SLA type. Then $\tau' \lhd \tau$.
\end{thm}
\begin{proof}
One can observe that one interval of $\tau$ will contain $K$ intervals of $\tau'$ with each interval providing $C$ computation time. Thus $\tau$ is satisfied.
\end{proof}

Next we present transformations that allow for missed allocation, but unlike arbitrary SLA modifications, These transformations provide a bound on the number of missed deadlines over a specific number of intervals, which deemed to be acceptable to satisfy a customer's SLAs. We begin by presenting the definition of bounds on missed deadlines.

\begin{define}
    $(C',T') \lhd_{a,b} (C,T)$ where $a$ is the bound on the missed deadlines over $b$ intervals of length $T$.
\end{define}

Theorems \ref{MultiplyByK} and \ref{TransferCSame} provide bounds on the number of missed deadline as we modify the allocation interval $T$.

\begin{thm}\label{MultiplyByK}
Let $\tau = (C,T)$ be an SLA type, and $\tau' = (C', T')$ be a host-provided SLA type, where  $T' = KT$ for some $K > 1$ then:
\begin{enumerate}
  \item If $0 \leq C' < K  \times  (C-1) + 1$ then $\tau' \lhd_{a,b}
  \tau$, where $a = K \times T$ and $b = K$. Moreover in such a case,
  $\alpha$ will miss at least one allocation deadline every $K$ intervals.
  \item For every $J \in \{1,\ldots,K-1\}$, if
  \begin{align*} & K  \times  (C-1) + (J-1)  \times  (T - (C-1)) + 1 \\
    &\qquad\leq C' < K  \times  (C-1) + (J)  \times  (T - (C-1)) + 1
  \end{align*}
	then $\tau' \lhd_{a,b} \tau$, where $a = (K - J)$ and $b = K$
  \item For $J = K$, if \[K  \times  (C-1) + (J-1)  \times  (T - (C-1)) +
      1 \leq C' \leq  T'\] then $\tau' \lhd \tau$.
\end{enumerate}
\end{thm}

\begin{proof} 
\noindent{\em Condition 1:} Since $T' = KT$, we have $K$
  intervals. No matter how the distribution of $C'$ is going to
  be over the $K$ intervals, it is always the case that the
  resource allocation will be less that the $KC$ units needed
  over the $K$ intervals. Thus we conclude that the schedule
  will always include at least one interval with a missed
  allocation.\\
\noindent{\em Condition 2:} Consider the left inequality in the conjunction,
  \textit{i.e.}, $K  \times  (C-1) + (J-1)  \times  (T - (C-1)) + 1 \leq
  C'$. Assume that there are $(K-J)$ unsatisfied intervals
  with at most $(K-J)  \times  (C-1)$ allocation units. Thus, there
  should be $J$ satisfied intervals containing at least $C' -
  (K-J)  \times  (C-1)$ allocation units. Therefore we have:
      $ C' >  J  \times  (T - (C-1))  + K  \times  (C-1)$.

  Since $ (K-J)  \times  T$ is the total time in all the satisfied
  intervals, it follows that the total time in the satisfied
  interval is strictly less than the allocations in the
  satisfied interval -- a contradiction. Therefore,
      $ C' <  J  \times  (T - (C-1))  + K  \times  (C-1) + 1.$

  Now, consider right inequality in the conjunction,
  \textit{i.e.}, $C' < K  \times  (C-1) + (J)  \times  (T - (C-1)) + 1$. If
  $C' < K  \times  (C-1) + (J-1)  \times  (T - (C-1)) + 1$, then there exists
  a schedule such that the number of satisfied interval is
  strictly less than $J$. Let $C' = K  \times  (C-1) + (J-1)  \times  (T -
  (C-1)) < K  \times  (C-1) + (J-1)  \times  (T - (C-1)) + 1$. We can simply
  distribute $C-1$ allocation units
      over $K$ intervals such that none of the intervals are
  satisfied. Furthermore, we distribute $T-(C-1)$ allocation
  units over $J-1$ windows, thus completely filling $J-1$
  intervals with $T$ allocation units.  Thus, we end up with at
  least $J-1$ satisfied intervals.\\
  \noindent{\em Condition 3:} To guarantee all intervals, in the
  worst case, we need to have $K-1$ intervals filled with $T$
  allocation units. In addition, we need to have at least $C$
  allocation units in the last interval. By substituting $K$ for $J$ in the above equation we get:
      $C \geq K \times  (C-1) + (K-1)  \times  (T - (C-1)) + 1$ therefore,
	$C \geq C + (K-1) \times T$.

\end{proof}

\begin{thm} \label{TransferCSame}
Let $\tau = (C, T)$ be an SLA type and $\tau' = (C, T')$
be a host-provided SLA type, where $(T+C)/2 < T' < T$ and $C
\leq T'$. If $m ={\rm lcm}(T,T')/T$, and $n = {\rm lcm}(T,T')/T'$ where lcm is
the least common multiple, then
\begin{enumerate}
  \item We can guarantee at least $s = n - m + 1$ satisfied
      intervals out of total $m$ intervals.
  \item We can guarantee at least $l =
      \lceil\dfrac{m}{(C+1)}\rceil$ satisfied intervals out
      of the total $m$ intervals.
\end{enumerate}
We can bound the number of missed deadlines every $m$ intervals to be $a = m - \max(s,l)$.
Therefore $\tau' \lhd_{a,b} \tau$ where $a = m - \max(s,l)$ and $b=m$.
\end{thm}

\begin{proof} 
\noindent{\em Condition 1:} Since $T > T'$, we observe that the number of
      satisfied intervals of length $T$ is at least equal to the number of
      completely overlapping intervals of length $T'$.
      Let $f(T,T')$ be the number of completely overlapped
      unique intervals of $\tau$ in $\tau'$, then \[f(T,T') =
      f(\frac{T}{\gcd(T,T')}, \frac{T'}{\gcd(T,T')})\] where
      $\gcd$ is the greatest common divisor of $T$ and $T'$.
      Thus $\frac{T}{gcd(T,T')}$ and
      $\frac{T'}{gcd(T,T')}$ are prime with respect to each
      other.
      Let $R = \{K \times T' \mod T\;|\;1 \leq K \leq T\}$, 
      then $|R| = T$. Furthermore $R = \{1, \ldots ,
      T\}$.\footnote{For simplicity, we choose to enumerate from 1 to $T$ instead of
      from 0 to $T-1$.}
      Since the remainders in $R$ are unique, let us mark
      the remainders on a circle starting from $1$ and
      ending at $T$. We observe that
      every remainder that is marked at the region starting
      from 1 to $\frac{T}{\gcd(T,T')}-\frac{T'}{\gcd(T,T')}$
      will not pass the cycle ending at $T$ because $T > T'$. In addition,
      the interval that starts at position $T$ will also not
      pass the cycle, therefore the total number of
      overlapping intervals is:
      \begin{align*}
	  &\frac{T'}{\gcd(T,T')} - \frac{T}{\gcd(T,T')} + 1 = n - m + 1. 
      \end{align*}\\
\noindent{\em Condition 2:}  Since $T'<T$, the first interval will always be satisfied. To bound the number
      of missed allocations, we assume an adversary whose
      purpose is to maximize the number of missed allocations by
      allocating the resource to intervals
      that are already satisfied. Under such conditions,
      we prove that every $C+1$ intervals of length $T$ will contain at
      least one satisfied interval.

Consider any schedule, assume it has a sub-sequence
      $S$ of $C+1$ unsatisfied intervals of length $T$ denoted
      by $T_1, \ldots , T_{c+1}$ where $T_1 = [t, t+T],
      \ldots, T_{c+1} = [t+CT, t']$. 
      Exactly $C$ intervals of length $T'$ are completely contained in $S$ denoted as $T'_1,
      \ldots, T'_c$. Let $T'$ start at $t_1$ and $T_c$ end
      at $t_2$. The total computation time in
      $[t_1,t_2]$ is at most equal to the total computation time in
      $[t,t']$.
      The total computation time scheduled in
      $[t_1,t_2] = C  \times  C = C^2$. Since all the
      intervals $T_i$ in $S$ are unsatisfied, the total
      computation time scheduled for each interval of length
      $T$ can be at most $C - 1$. Therefore the total computation time in $[t,t'] \leq (C+1)  \times  (C
      - 1) = C^2 - 1$. Contradiction. $S$ must contain some satisfied intervals.
      \begin{figure} [htp] \label{Cplus1windows}
       \centering
      \includegraphics[width=0.6\textwidth]{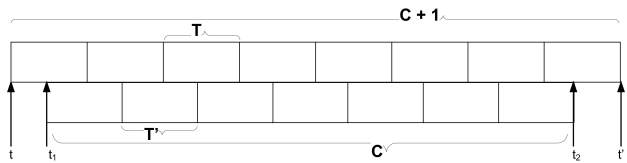}
      \caption{Adversary trying to miss maximum possible deadlines} \label{adversary}
      \end{figure}
To generalize, we have $m$ intervals of length $T$. Since every
$C+1$ intervals of length $T$ will contain at least one
satisfied interval, we can bound the number of missed
allocations to be at most equal to $m - \lceil\frac{m}{(C+1)}\rceil$.
\end{proof}

We also define a two step transformation of an SLA type by applying the transformation in Theorem \ref{MultiplyByK} the transformation in Theorem \ref{TransferCSame}.

\begin{thm} \label{Compose1}
Let $\tau_1 = (C_1, T_1)$,  and $\tau_2 = (C_2,T_2)$ such that $\tau_2 \lhd_{a,b} \tau_1$ by applying the transformation in Theorem \ref{MultiplyByK}.
Let $\tau_3 =(C_3,T_3)$ such that $\tau_3 \lhd_{x,y} \tau_2$ by applying the transformation in using Theorem \ref{TransferCSame}.
Then $\tau_2 \lhd_{c,d} \tau_1$ where $c=(b \times x+(y-x) \times a)$ and $d=(b \times y)$.
\end{thm}

\begin{proof}
$\tau_2 \lhd_{a,b} \tau_1$ will miss at most $a$ allocations over $b$ intervals.
$\tau_3 \lhd_{x,y} \tau_2$ will miss at most $x$ allocations over $y$ intervals.
Every missed allocation in $\tau_2$ corresponds to the failure of satisfying an entire window $b$ in $\tau_1$, and every satisfied window in $\tau_2$ corresponds to missing at most $a$ allocations in $\tau_2$. Thus at most, the total number of missed allocations over a window $d=(b \times y)$ is the sum of all possible missed allocations
$c=(b \times x +(y-x) \times a)$.
\end{proof}

\begin{thm}\label{same}
Applying the transformations in Theorems \ref{MultiplyByK} and \ref{TransferCSame}, in this order, is equivalent to applying the transformations in Theorems \ref{TransferCSame} and \ref{MultiplyByK}, in this order. That is, the two transformations commute.
\end{thm}

\begin{proof}
Theorem \ref{Compose1} highlights results of applying Theorem \ref{MultiplyByK} followed by Theorem \ref{TransferCSame}. We would like to show that the results for Theorem \ref{TransferCSame} followed by Theorem \ref{MultiplyByK} are equal.

$\tau_2 \lhd_{a,b} \tau_1$ will miss at most $a$ allocations over of $b$ intervals.
$\tau_3 \lhd_{x,y} \tau_2$ will miss at most $x$ allocations over $y$ intervals.
Every missed allocation in $\tau_2$ corresponds to the failure of satisfying an entire window $b$ in $\tau_1$, and every satisfied window in $\tau_2$ corresponds to missing at most $a$
allocations in $\tau_2$. Thus at most, the total number of allocations over a window $d=(b \times y)$ is the sum of all possible missed allocations $c=(b \times x +(y-x) \times a)$.
\end{proof}

\section{SLA Model: Fluidity} \label{sec:SLAModel}
As we alluded before, we believe that a periodic resource allocation
  model is appropriate for expressing SLAs in an IaaS setting. Thus,
  in this section we extend the periodic SLA model for the purpose of expressing general SLAs of IaaS customer workloads -- which may not be inherently ``real time''. In particular, we extend the SLA model to allow for the modeling of ``fluid'' workloads.

A {\em fluid workload} is one that requires {\em predictable} periodic
  allocation of resources ({\em i.e.} not best effort), but has
  flexibility in terms of how such periodic allocations are
  disbursed ({\em i.e.} not real-time). For instance, a fluid workload may specify a periodic
  need for resources as long as the disbursement of these cycles is
  guaranteed over some acceptable range of periods. For example, a
  fluid workload may specify the need for 10K cycles per second as
  long as these cycles are disbursed over a fixed period in the range
  between 100msec and 10 secs. Thus, a disbursement of 1K cycles every
  100 msecs is acceptable as is a disbursement of 100K cycles every 10
  secs. But, a disbursement of 200K cycles every 20 secs would be
  unacceptable as it violates the upper bound imposed on the
  allocation period, and so would an allocation of 100 cycles every 10
  msecs as it violates the lower bound. 

It is important to highlight that our periodic allocation is less stringent than what a
  ``real-time'' workload may require, but more stringent than what
  typical IaaS virtualization technologies are able to offer. Unlike real-time systems, there is no notion of deadlines, but rather an expectation of resource allocations at a prescribed predictable rate.

\begin{define}
An SLA $\tau$ is defined as a tuple of natural numbers
  $(C,T,T_l,$ $T_u,D,W)$, such that $0 < C \leq T$, $T_l \leq T \leq
  T_u$, $D \leq W$, and $W \geq 1$, where $C$ denotes the resource
  capacity supplied or demanded during each allocation interval $T$,
  $T_l$ and $T_u$ are lower and upper bounds on $T$, and $D$ is the
  maximum number of times that the workload could tolerate missing an
  allocation in a window consisting of $W$ allocation
  intervals.
\end{define}

According to the above definition, an SLA of type $(C,T,T_l,T_u,
  D,W)$ represents a fluid workload which requires an allocation of
  $C$ every interval $T$, where $T$ can vary between $T_l$ and $T_u$
  as long as the ratio $C/T$ is consistent with the original SLA type.
  The following are illustrative examples

An SLA of type $(2,4,2, 8, 0,1)$ represents a fluid workload that
  demands a capacity $C=2$ every allocation interval $T = 4$, however
  the original SLA would still be satisfied if its gets a capacity
  $C'=4$ every allocation interval $T'=8$ since the ratio $C'/T'$ is
  equal to $C/T$.

An SLA of type $(2,4,2, 8, 1,5)$ is similar in its demand
  profile except that it is able to tolerate missed allocations as long as there are no more than $D=1$
  such misses in any window of $W =5$ consecutive allocation periods.

\noindent {\bf Fluid Transformations:} In addition to the
  transformations defined in Section \ref{sec:basics}, we introduce the following
  transformation for fluid workloads.

  \begin{thm}
  A fluid SLA of type $\tau' = (C', T', T_l, T_u, D, W)$ satisfies an SLA of type $\tau = (C, T, T_l, T_u,$ $D, W)$ if $T_l \leq
    T' \leq T_u$ and $C' = \lceil C  \times  T'/T\rceil$.
  \end{thm}

\begin{proof}
$C' = \lceil C  \times  T'/T\rceil$ implies the ratio of $C'/T' \geq C/T$. $T_l \leq T' \leq T_u$ implies that $T'$ is an acceptable allocation period for fluid workload (based on the definition). Therefore $\tau' = (C', T', T_l, T_u, D, W)$ satisfies $\tau$
\end{proof}

\section{\MORPHOSYS: The Framework}\label{sec:DynamicService}

\noindent
  In this section, we exploit the theorems defined in Sections \ref{sec:basics} and \ref{sec:SLAModel} as building blocks in our \morphosys\ framework.
We consider an IaaS setting consisting of any number of homogeneous
  instances (servers), to which we refer as ``Physical Machines''
  (PM).\footnote{Again, we emphasize that while we present our framework in the context of computational supply and demand -- using   terminologies such as physical and virtual machines -- \morphosys\   is equally applicable to other types of resources.}
  Each workload
  (served with a virtual machine instance) is characterized by an SLA
  that follows the definition above -- namely $\tau = (C, T, T_l, T_u,
  D, W)$. The \morphosys\ colocation framework consists of two major
  services: a Workload Assignment Service (WAS) and a Workload
  Repacking Service (WRS). WAS assigns workloads to PMs in an on-line
  fashion using a prescribed assignment policy. WRS performs workload
  redistribution across PMs to optimize the use of cloud resources.

\subsection{Workload Assignment Service (WAS)}\noindent
Figure \ref{fig:WASFlowchart} provides an overview of the main
  elements of WAS. WAS is invoked upon the arrival of a request for a
  Virtual Machine (VM) allocation, in support of a workload specified
  by an SLA. The WAS service uses one of two heuristics to select the
  PM that could potentially host the VM: First Fit (FF) and Best Fit
  (BF). FF assigns the VM to the first PM that can satisfy the VM's
  SLA, whereas BF assigns the VM to the fullest -- most utilized -- PM
  that can still satisfy the VM SLA.

If it is not possible for WAS to identify (using FF or BF) a PM
  (currently in use) that could host the newly-arriving VM, then WAS
  attempts to rewrite the SLA of the VM (safely) in the hopes that it
  would be possible to assign the VM (subject to the transformed SLA)
  to an existing PM. To do so, WAS proceeds by generating a safe SLA
  transformation and attempts to use either FF or BF to find an
  assignment. This process is repeated until either one of the safe
  SLA transformations results in a successful assignment of the VM to
  a PM, or WAS runs out of possible safe SLA transformations. In the
  latter case, WAS may invoke the WRS repacking service to repack
  already utilized hosts in an attempt to assign the workload, or
  alternatively WAS can simply instantiate a new PM to host the
  newly-arriving PM.

In the worst case, the complexity of WAS is $O(k  \times  n)$ where $k$ is
  the largest number of possible task transformations per task across all possible tasks, and $n$ is the
  number of hosts in the system. Although its possible for $k$ to be the dominant factor in the complexity, based on experimental observations (inferred by traces from real-workloads), $k << n$, which implies that in practice, WAS scales linearly with the number of hosts.

\subsection{Workload Repacking Service (WRS)}\noindent
Repacking is an essential service that allows the
  remapping/reclustering of workloads. This service is needed because
  IaaS environments may be highly dynamic due to the churn caused by
  arrival and departure of VMs, and/or the need of customers to change
  their own resource reservations. Over time, such churn will result
  in under-utilized hosts which could be managed more efficiently if
  workloads are repacked.

\begin{figure} [htp]
\begin{center}
  \includegraphics[width=0.47\textwidth]{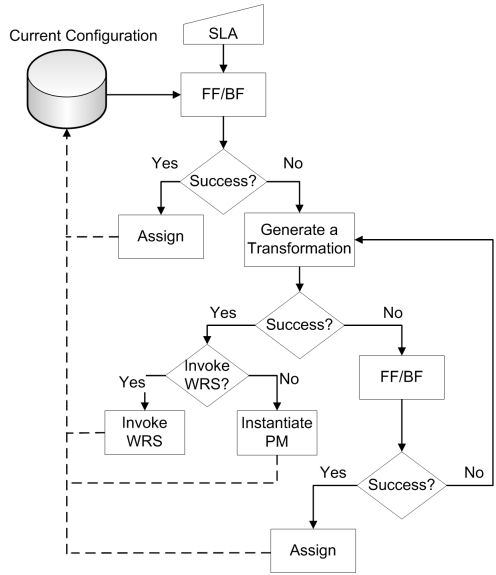}
  \caption{\label{fig:WASFlowchart} The WAS Component of \morphosys.}
\end{center}
\vspace{-0.3in}
\end{figure}

\noindent {\bf Repacking Heuristic:} Remapping a set of workloads to
  multiple hosts efficiently is the crux of the problem. We say
  ``efficient'' as opposed to optimal because multi-processor
  real-time scheduling has been shown to be NP-Hard
  \cite{garey1979computers} (and our problem by reduction is also
  NP-Hard), and thus we resort to heuristics. Many such heuristics (approximations) have
  been proposed in the literature ({\em e.g.}, based on the use
  of a bin packing or greedy strategy) \cite{Baruah90algorithmsand,Londono2008NETEMBED}.

Our safe SLA transformations provide us with another degree of freedom ({\em i.e.}, another dimension in the search space): Rather than finding the best packing of a set of tasks with fixed SLA requirements (the original NP-hard problem), we have the flexibility of safely manipulating the SLAs with the hope of achieving a better packing. Towards that end, we have implemented heuristic algorithms that utilize Breadth First Search (BFS) and Depth First Search (DFS) techniques to explore the solution search space.

\begin{algorithm}
 \algsetup{linenosize=\scriptsize}
  \scriptsize
   \caption{ Repack(HostList): WRS Repack Heuristic}
\label{alg:repackHeuristic}
\begin{algorithmic}[1]
	\STATE {BestSolSoFar $\leftarrow$ size(HostList)}
	\FORALL{Host in HostList} 
		\FORALL{Tasks in Host} 
			\STATE{TransTaskList $\leftarrow$ Gen\_Trans(Task)}
			\COMMENT{Generate a list of transformations for the particular task}
			\STATE CandTaskList[TaskID] $\leftarrow$TransTaskList
		\ENDFOR
	 \ENDFOR
	\STATE SortTaskList(CandTaskList) 	\COMMENT{Sort Tasklist ascending}
	\STATE Tree $\leftarrow$ EmptyNode
	\REPEAT
		\STATE {Task $\leftarrow$ SortTaskList.Remove(0)}
		\STATE AddNodeToTree(Task) \COMMENT{Add original task and transformations as nodes in the tree.}
        \STATE UtilizationOk(Tree) \COMMENT{Prunes tree branches}
    \UNTIL{SortTaskList $= Empty$}
	\STATE BestSol = GetBestSolution(Tree)
	\STATE return BestSol
\end{algorithmic}
\end{algorithm}

Our (BFS or DFS) heuristic starts with a preprocessing stage, in which we generate all possible transformations for each task (using our arsenal of safe transformations). Next, it proceeds by setting up the search space (tree or forest) of all the alternative task sets that could be colocated. Finally, it proceeds to explore that search space with the aim of finding a feasible transformation.

In the worst case, our heuristic may end up searching the entire solution space, which is obviously impractical. To manage the exponential
  nature of the search space, our heuristic utilizes two optimization (pruning) strategies.

Our first optimization strategy adopts an early-pruning approach: at each stage of our search, if the aggregate utilization (demanded SLA) of the tasks under consideration thus far (whether these tasks are transformed or not) is greater than the capacity of the host (supplied SLA), then we prune that branch of the tree on the assumption that a feasible solution cannot exist down that path. We initially set the ``best'' solution to be the total number of hosts used prior to repacking.

Our second optimization adopts a smaller-degree-first approach: we build the search space (tree) by greedily starting with tasks that have the smallest number of transformations. This ensures that when pruning is applied (using the above strategy) we are likely to maximize the size of the pruned subspace. This optimization strategy has been shown to be quite effective in reducing the solution search space for network embedding problems \cite{Londono2008NETEMBED}. For practical purposes, we set an upper-bound on the execution time of the repacking heuristic, which we take to be 5 minutes for services operating on hourly ``pay-as-you-go'' reservations.

Procedure \ref{alg:repackHeuristic} illustrates our repacking process. Given a set of candidate hosts, we proceed to generate transformations for each task on a particular host (lines 2-7). We sort the candidates based on the number of transformations, which is a precursor for applying the first optimization strategy. Lines 9-14 construct the tree by adding each task and its transformations at the leafs of the tree, Then prunes the branches of the tree where solutions exceed the current best solution. Finally, the best solution is returned (lines 15-16).

\noindent {\bf Repacking Policies:} Our WRS service could be
  instantiated based on one of three possible repacking policies: No
  Repacking (NR), Periodic Repacking (PR), and Forced Repacking
  (FR). NR is used to disable WRS, PR allows repacking to run at
  designated epochs/periods based on a system defined parameter. FR
  allows repacking to be applied in a ``on line'' fashion (triggered
  by the WAS).

  \begin{algorithm}
   \algsetup{linenosize=\scriptsize}
    \scriptsize
  \caption{SelectHosts(): WRS Host Selection}
  \label{alg:repack}
  \begin{algorithmic}[1]
   \IF {NM}
  	\FORALL{i in AllHosts}
  		\STATE{Repack($i$)} \COMMENT{repack a single host}
  	\ENDFOR
    \ELSE
  	\IF {CM}
  		\STATE{Hostlist $\leftarrow$ GetCandidateHosts()}
  		\COMMENT{Constrained migration: find the set of suitable candidate hosts}
  	\ELSE
  		\STATE{Hostlist $\leftarrow$ AllHosts}
  		\COMMENT{Unconstrained migration select all hosts}
  	\ENDIF
  	\STATE{Repack(Hostlist)}
   \ENDIF
  \end{algorithmic}
  \end{algorithm}

\noindent {\bf Migration Policies:} The effectiveness of the repacking
  policy depends on the ability to migrate workloads from one host to
  another. However, adding hosts increases the total number of
  workloads to be repacked which in turn results in an increase in the
  total service turnaround time. Thus we model three types of
  migration policies: No Migration (NM), Constrained Migration (CM),
  and Unconstrained Migration (UM). NM policy allows for repacking on
  the condition that workloads will not migrate from the host to which
  they are assigned. This approach will naturally consider {\it one}
  host at a time, and is suitable for running WAS in an ``online''
  fashion. In particular, under a NM policy, as the framework receives a new request for task assignment, it attempts to apply transformations on tasks assigned to that particular host as well as the candidate request. The goal of applying the transformations is to assign all tasks to the host. CM and UM policies allow for workloads to migrate from one
  host to another as long as it results in a more efficient repacking
  of these workloads. This is suitable when WAS is run in an
  ``offline'' fashion. The difference between CM and UM is in the host
  selection criteria: UM considers all system hosts,
  whereas CM considers hosts that satisfy a host selection
  condition.

\noindent {\bf Host Selection Condition:} A host is a candidate for
  repacking if it satisfies a condition on its utilization. Let
  $0<\phi <1$ be the average host utilization, which we define as
$\phi = \frac{\sum_{i = 1}^{n} u_i}{n}$, where $u_i >0$ is the
  utilization of host $i$. Furthermore, let $\omega =\sum_{i = 1}^{k}
  C_i/T_i$ be the sum of the utilizations of the workloads on a
  specific host (based on the original workload SLA and not the transformed workload SLA). A host is a candidate for repacking if $\phi - \omega
  \geq \epsilon$, where $0 < \epsilon \leq 1$ is a tunable parameter
  (CM reduces to UM, when $\epsilon = 0$). The logic for all WRS variants is shown in Procedure \ref{alg:repack}.

\section{\MORPHOSYS: Evaluation}\label{sec:experimental}\noindent
In this section we present results from extensive experimental
  evaluations of the \morphosys\ framework.  Our main purpose in doing
  so is to establish the feasibility of our proposed service by: (1)
  comparing the schedulability of QoS workloads, with and without
  applying our safe SLA transformations, (2) evaluating the effect
  from using different migration policies on the efficiency of
  colocation, (3) evaluating the effect of changes in the mix of fluid
  and non-fluid workloads, and (4) evaluating the effect of changing
  the flexibility of fluid workloads on the efficiency of
  colocation.

 \noindent{\bf Simulation Setup:} Our setting is that of a cloud storage service used to host a large number of streaming servers. This setting is general enough to represent different forms of applications, such as a cloud content provider streaming and other multimedia services. Typically for such applications, the disk I/O constitutes the bottleneck of the overall system performance \cite{cherkasova2003building, barker2010empirical}. The maximum throughput delivered by a disk depends on the number of concurrent streams it can support without any SLA violation.

To drive our simulations, we utilize a collection of video traces from \cite{Auwera2008Traffic}. We assume that the underlying system of the provider is a disk I/O system that serves requests of different streaming servers using a fixed priority scheduling algorithm, which we take to be Rate Monotonic. The usage of fixed priority algorithms for disk scheduling was suggested by Daigle and Strosnider \cite{daigle1994disk}, and Molano {\em et al} \cite{molano2002real}.

The video traces \cite{Auwera2008Traffic} provide information about the frames of a large collection of video streams under a wide range of encoder configurations like H.264 and MPEG-4. We conducted our experiments with a subset of 30 streams, with HD quality, and a total duration of one hour each. We initially identify the period for serving a video stream request as the period of the {\texttt I} frames ({\em a.k.a.}, Group of Pictures, or GoPs). Overall, there were three unique periods in our collection of video traces.

We model the SLA associated with each stream as follows: The SLA
  specifies a periodic (disk I/O) demand $C$ over a periodic
  allocation time $T$. For a given stream, the periodic demand $C$ is set as follows:

$$\mbox{\fontsize{9}{9}\selectfont $\displaystyle C = \frac{max(\sum^{n-1}_{i=0} b_i)}{\theta  \times  T}$}$$

\noindent where $b_i$ is
  the volume in bytes of the stream in interval $[i  \times  \theta  \times  T,
  (i+1)  \times  \theta  \times  T]$.The allocation period $T$ is set to be equal to $\theta  \times  T'$, where
  $T'$ is one of the three unique periods in our video traces and
  $\theta$ ($\theta \geq 1$) models the tolerance of the client (the
  recipient of the stream) to burstiness in the allocation over
  time. In particular, for any given value of $\theta$, it is assumed that the client is able to buffer (and hence absorb) up to  $T = \theta  \times  T'$ seconds of the stream ({\em i.e.}, $\theta$ GoPs). 
A large value for $\theta$ implies that the allocation is
  over a large number of GoPs, and hence a tolerance by the client for
  a bursty disbursement of periodic allocation. A small value for
  $\theta$ specifies a smoother disbursement over time.  Each client
  request specifies a value for the parameter $\theta$ which is chosen
  at random between a lower bound $\beta$ and an upper bound
  $\gamma$. In our experiments we set $\beta = 1$ and $\gamma = 10$.
To model the level of fluidity (flexibility) of an SLA, we allow the
  period $T$ to range from $(\theta - \sigma)  \times  T'$ to $(\theta +
  \sigma)  \times  T'$, where $\sigma$ ($\theta \geq \sigma \geq 0$)
  determines the allowable deviation from the nominal allocation
  period. A non-fluid SLA is one where $\sigma = 0$.

We model churn in the system as
  follows. Client arrivals (requests for streams) are Poisson
  (independent) with a rate $\lambda$. Poisson arrival processes for
  VoD have been observed in a number of earlier studies ({\em e.g.},
  \cite{VelosoAlmeidaMeiraBestavrosJin:ton06}).  A client's session
  time is set to be the length of the entire stream served to the
  client.  The specific stream requested by the client is chosen
  uniformly at random from among all streams in the
  system. Experiments with skewed distributions ({\em e.g.}, Zipf) and inhomogeneous poisson process
  resulted in results that are similar to those obtained using exponential
  preference,\footnote{This is expected given the relatively similar
  lengths of the streams in the trace.} and thus are not reported.

In our experiments, our purpose is to evaluate the efficiency of computing a colocation configuration for our workloads, as opposed to the performance of system deployment. Thus,  to measure the efficiency of a colocation strategy
  $X$, we report the {\em Colocation Efficiency} (CE), which is
  defined as follows:

$$\mbox{\fontsize{9}{9}\selectfont $\displaystyle CE = 1 - \frac{W(X)}{W(FF)}$}$$

\noindent where $W(X)$ is the amount of wasted
  (unallocated) resources when colocation strategy $X$ is used, and
  $W(FF)$ is the measure of wasted resources when our baseline
  First-Fit (FF) strategy is used. Thus, CE can be seen as the degree
  to which a strategy is superior to FF (the reduction in wasteds
  resources relative to FF, which according to theoretical bounds \cite{Liu1973Scheduling}, can be up to 30\%). All CE values are
  reported with 95\% confidence.

 \begin{figure} [htp]
\begin{center}
  \includegraphics[width=0.60\textwidth]{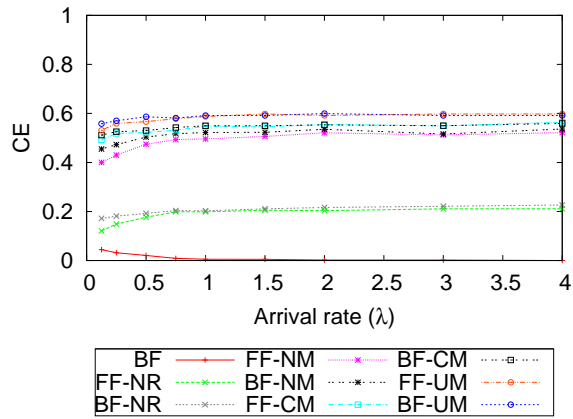} \\
	  \caption{Colocation Efficiency: Baseline Results}\label{fig:packingratio}
\end{center}
\vspace{-0.2in}
\end{figure}

\noindent{\bf Relative Performance of Various Strategies:} Recall that
  the assignment of an incoming workload (request) is done using WAC,
  which attempts various SLA transformations on an incoming request
  until the potentially transformed request is possible to assign to a
  host (disk) using either First-Fit or Best-Fit.

In a first set of experiments, we compared the performance of WAC with
  No Repacking under both FF and BF (namely FF-NR and BF-NR) to that
  of the plain FF and BF heuristics ({\em i.e.}, without attempting
  any SLA transformations). Figure \ref{fig:packingratio} shows the
  results we obtained when varying the arrival rate ($\lambda$) for
  the different packing strategies: BF, FF-NR, and BF-NR.

In general, the performance of BF is only marginally better than FF,
  whereas both FF-NR and BF-NR show measurable (up to 20\%)
  improvement over both FF and BF, with BF-NR performing slightly
  better than FF-NR. These results suggest that there is a measurable
  improvement in colocation efficiency even when minimal SLA
  transformations are allowed (namely the transformation of the SLA of
  the incoming request only).

  \begin{figure}[htp]
  \begin{center}
  \begin{minipage}{0.42\linewidth}
  	\centerline{\small ~~~~~~(a)}
  	\includegraphics[width=1\textwidth]{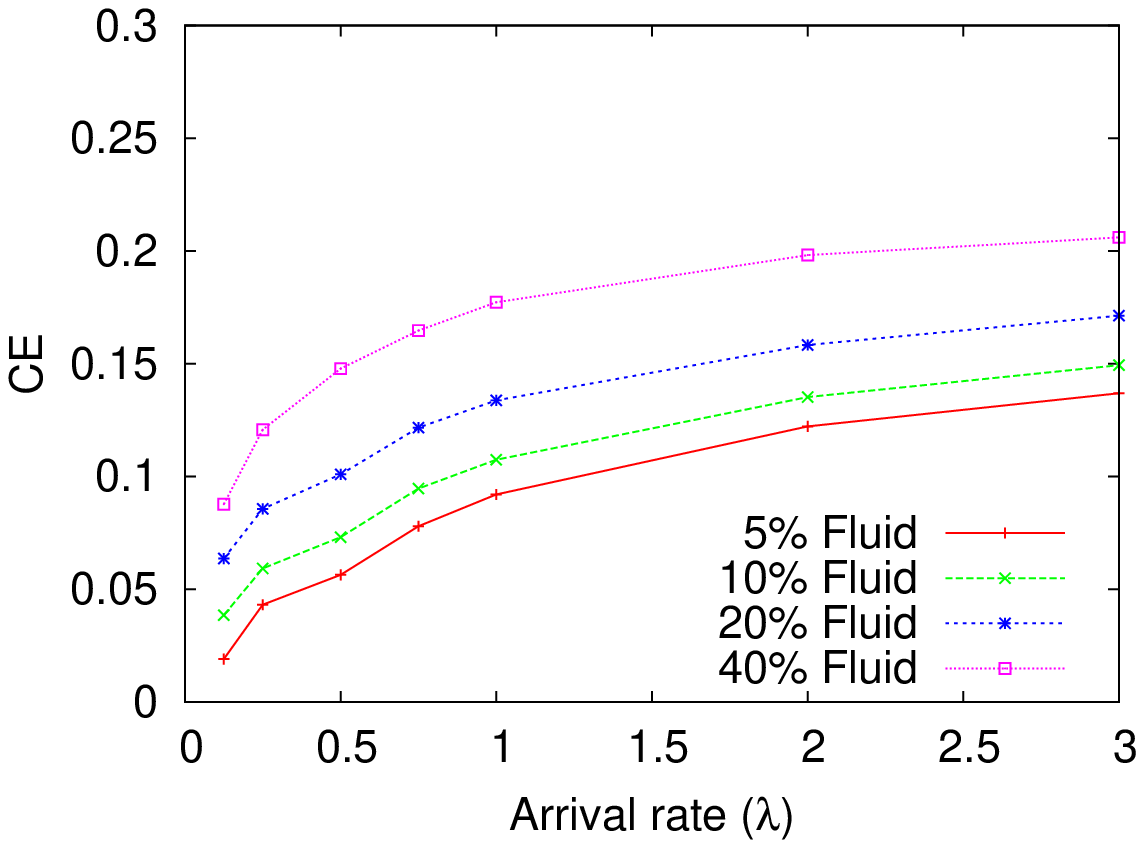}
  \end{minipage}
    \begin{minipage}{0.08\linewidth}
     $\quad$
    \end{minipage}
  \begin{minipage}{0.42\linewidth}
  \centerline{\small ~~~~~~(b)}\includegraphics[width=1\textwidth]{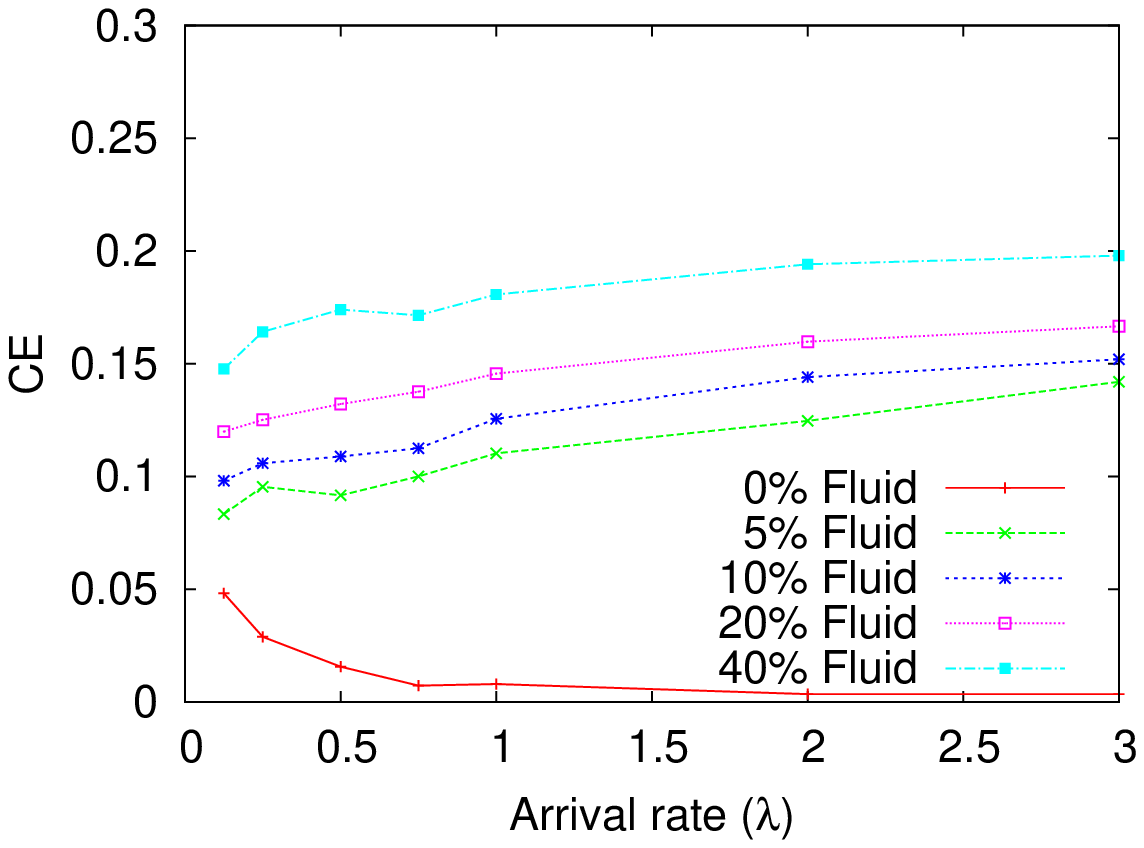}
  \end{minipage}
  \caption{Effect of fluid SLAs when only fluid transformations are allowed, $\sigma = 1$, (a) FF-NR. (b) BF-NR.}
  \label{fig:FluidTransSavings}
  \end{center}
  \vspace{-0.3in}
  \begin{center}
    \begin{minipage}{0.42\linewidth}
     \centerline{\small ~~~~~~(a)}
  \includegraphics[width=1\textwidth]{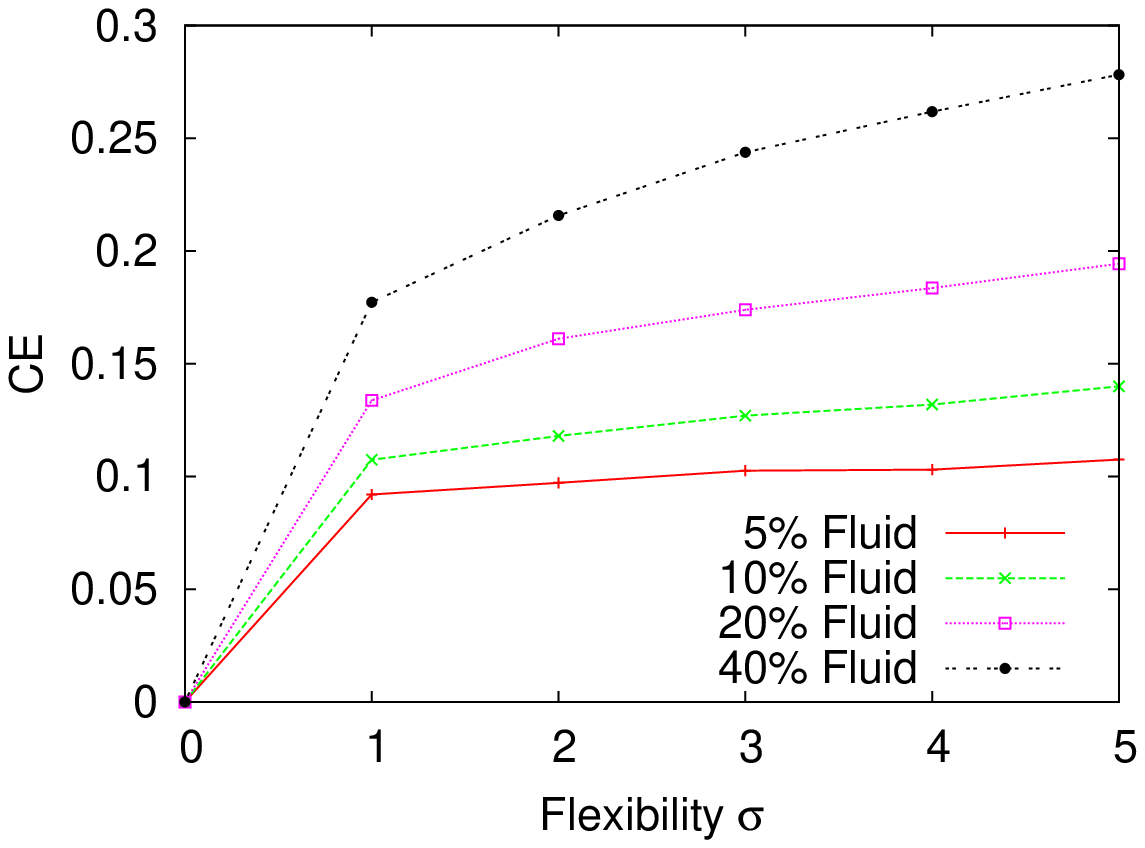}
    \end{minipage}
        \begin{minipage}{0.08\linewidth}
         $\quad$
        \end{minipage}
    \begin{minipage}{0.42\linewidth}
    \centerline{\small ~~~~~~(b)}\includegraphics[width=1\textwidth]{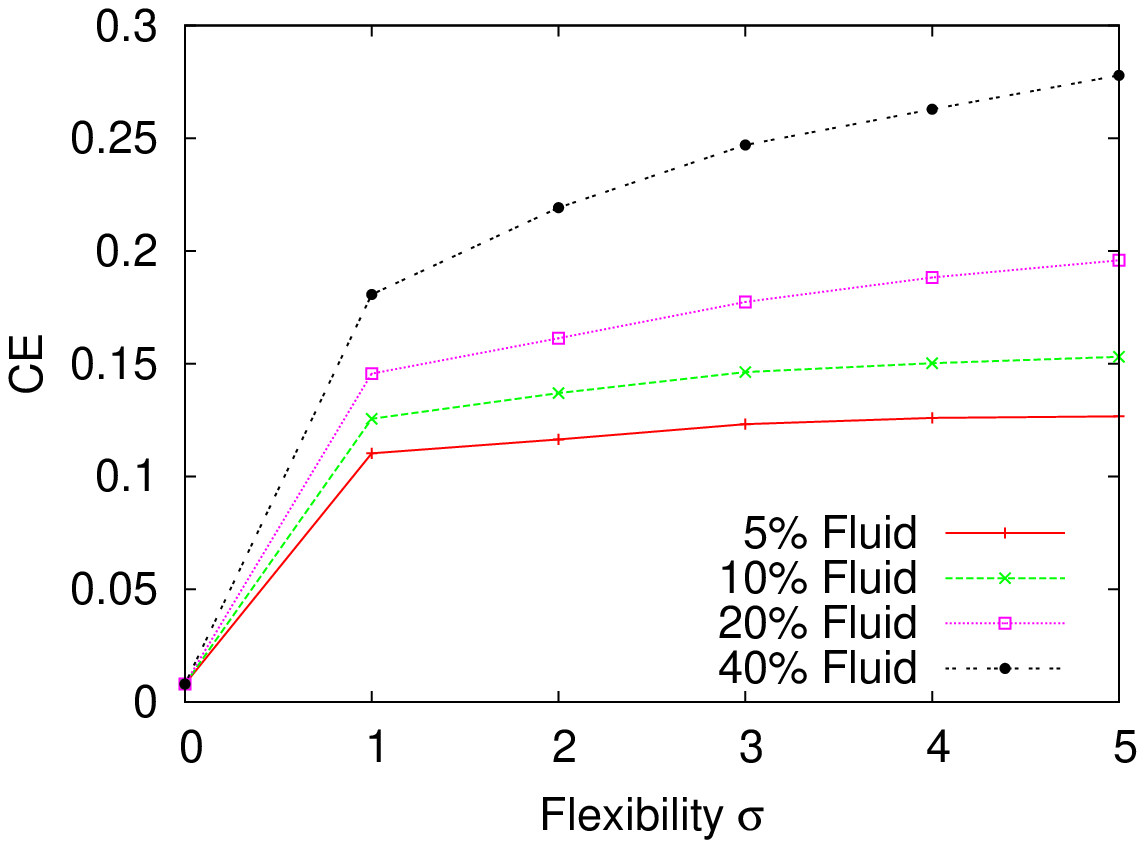}
    \end{minipage}
    \caption{ \label{fig:FluidTransStretchSavings} Effect of fluidity level when only fluid
    transformations are allowed, $\lambda = 1$, (a) FF-NR. (b) BF-NR.}
\end{center}
\vspace{-0.2in}
\end{figure}

To evaluate the benefit from repacking and migration, we ran a similar
  experiment with the repacking policy set to Forced Repacking
  (FR). Figure \ref{fig:packingratio} shows the measured CE values for
  different arrival rates ($\lambda$) and different repacking
  strategies. Our ``online'' repacking strategies with no migration --
  namely FF-NM and BF-NM -- improved colocation efficiency
  significantly. For lower arrival rates, CE was around 0.4, implying
  a 40\% reduction in wasted (unallocated) resources compared to
  FF. For moderate and higher arrival rates, the reduction is more
  pronounced around 50\%.

  \begin{figure}[htp]
  \begin{center}
    \begin{minipage}{0.42\linewidth}
     \centerline{\small ~~~~~~(a)}\includegraphics[width=1\textwidth]{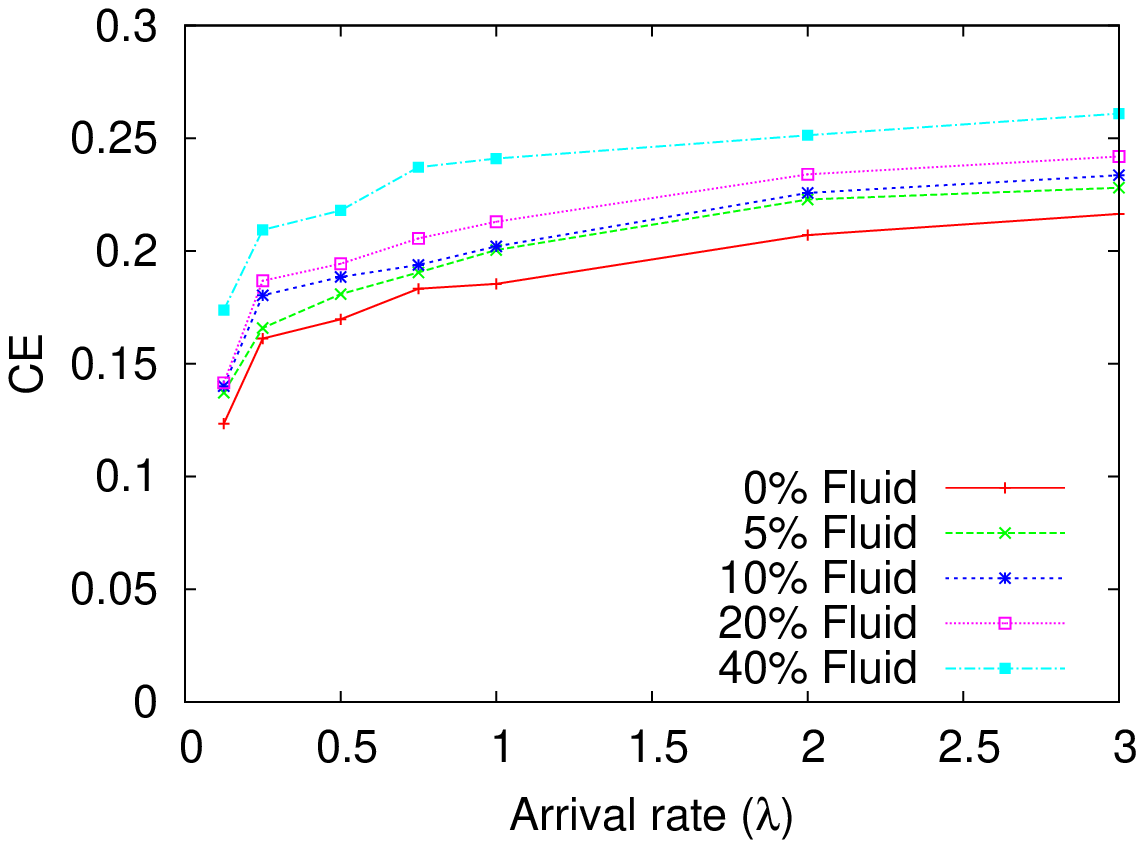}
    \end{minipage}
            \begin{minipage}{0.08\linewidth}
             $\quad$
            \end{minipage}
    \begin{minipage}{0.42\linewidth}
    \centerline{\small ~~~~~~(b)}\includegraphics[width=1\textwidth]{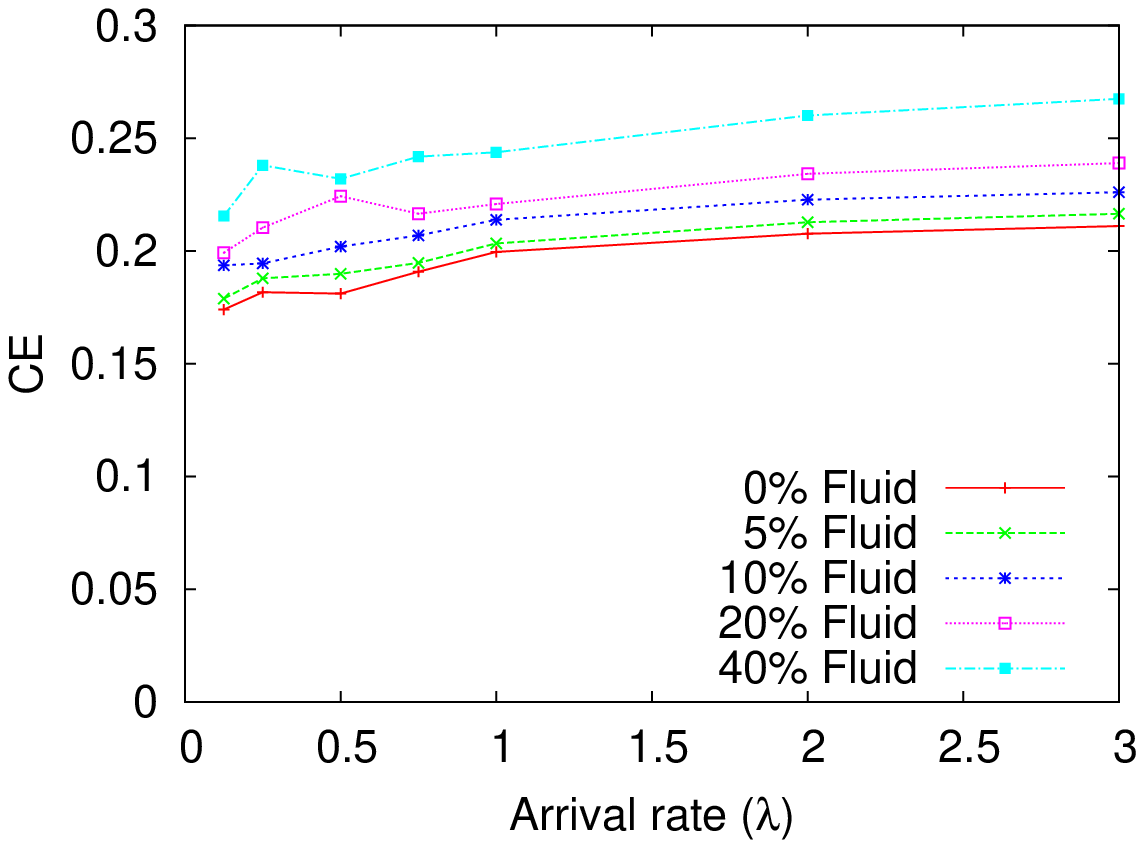}
    \end{minipage}
    \caption{\label{fig:BothTransSavings} Combined benefit from applying both fluid and non-fluid transformations, (a) FF-NR.  (b) BF-NR.}
    \begin{minipage}{0.42\linewidth}
     \centerline{\small ~~~~~~(a)}\includegraphics[width=1\textwidth]{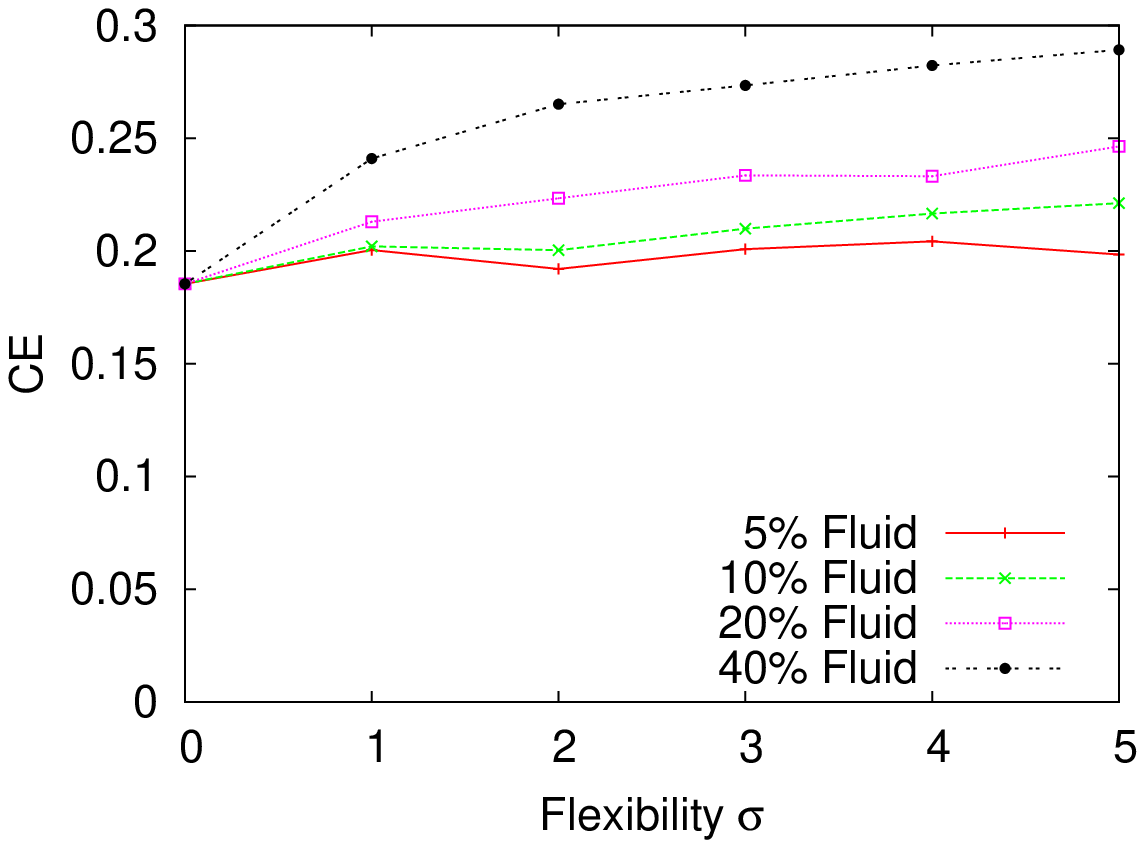}
    \end{minipage}
            \begin{minipage}{0.08\linewidth}
             $\quad$
            \end{minipage}
    \begin{minipage}{0.42\linewidth}
    \centerline{\small ~~~~~~(b)}\includegraphics[width=1\textwidth]{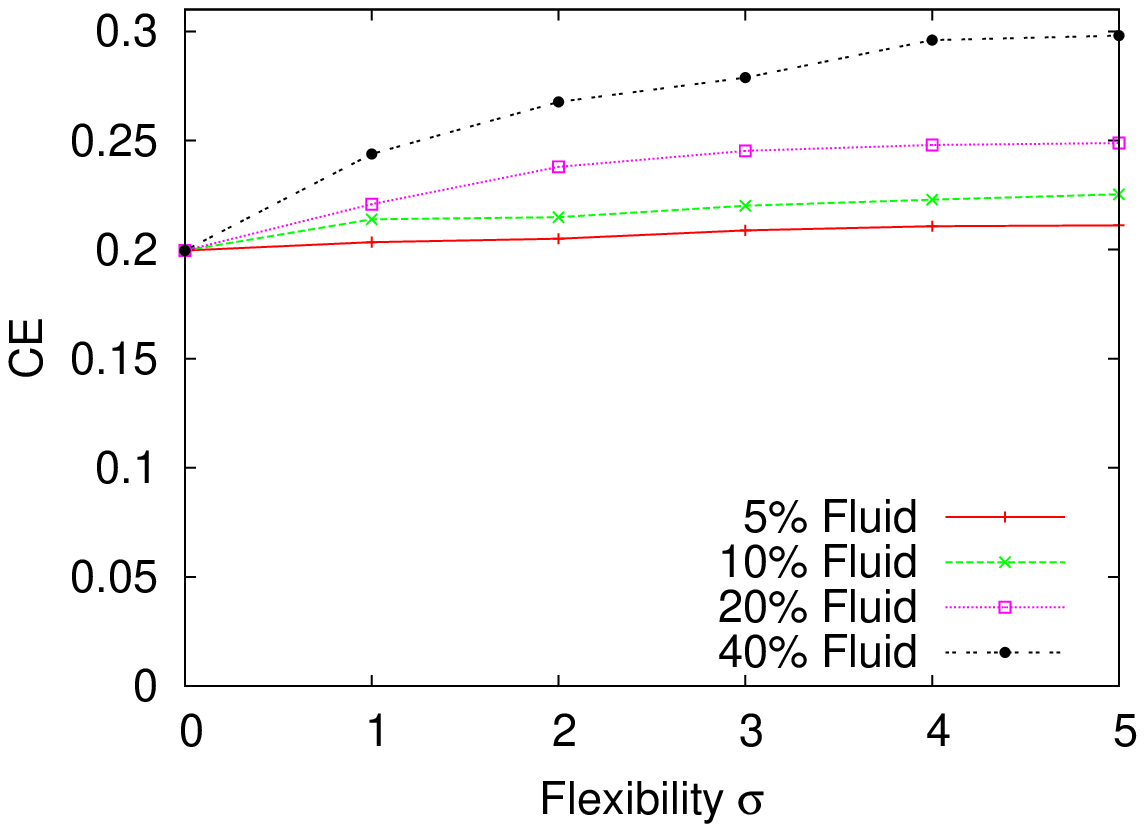}
    \end{minipage}
    \caption{\label{fig:BothTransStretchSavings}Effect of fluidity level when both fluid and non-fluid transformations are allowed, $\lambda = 1$: (a) FF-NR (b) BF-NR.}
  \end{center}
  \vspace{-0.3in}
   \end{figure}

Figure \ref{fig:packingratio} also shows results of experiments in
  which various migration policies are enabled -- namely Constrained
  Migration (CM) and Unconstrained Migration (UM). Both approaches
  result in better performance compared to NM approaches, yielding CE
  values between 0.55 and 0.6.

To summarize, this initial set of experiments suggests that, even in
  the absence of any fluidity in the workload (SLA flexibility), a
  reduction of up to 60\% in the wasted resources is to be expected
  through the use of SLA transformations and repacking.

\noindent{\bf Benefit from Fluid Transformations:} To measure the
  effect of fluidity on the overall colocation efficiency, we
  performed experiments using the same setting as before, while
  allowing a certain percentage of the requests to be fluid (with
  $\sigma = 1$), and only allowing fluid transformations to be
  applied. In other words, non-fluid workloads were not subjected to
  any transformations.

Figures \ref{fig:FluidTransSavings} (a) and (b) show the results we obtained using both FF-NR, and BF-NR, respectively, for various mixes of fluid and non-fluid workloads. As the result suggests, having a mix with even a small percentage of fluid workloads results in improvements (up to 20\%) that are comparable to what we obtained when transformation of non-fluid workloads was allowed with no repacking (cf. Figure \ref{fig:packingratio}).

In the previous experiment, we fixed the fluidity level ($\sigma = 1$)
  and studied the effect of changes in the mix of fluid versus
  non-fluid SLAs. Figure \ref{fig:FluidTransStretchSavings} (a) and (b) show results of additional experiments in which we changed the
  level of fluidity (the parameter $\sigma$) while keeping the value
  of $\lambda = 1$, for various mixes of fluid and non-fluid
  workloads. The results show that only small levels of flexibility
  ($\sigma < 2$) provided most of the achievable improvements when
  only fluid transformations are considered.

\noindent{\bf Combined Benefit from Fluid and non-Fluid
 Transformations:} Fixing the level of fluidity to a small value
 ($\sigma = 1$), Figures \ref{fig:BothTransSavings} (a) and (b) show
 results from experiments in which all possible transformations are
 allowed in a No-Repacking setting ({\em i.e.}, FF-NR  and BF-NR) for
 different workload mixes.  Figures \ref{fig:BothTransStretchSavings} (a) and (b) show results of experiments with similar settings -- all possible transformations are allowed in a No-Repacking setting -- with $\lambda = 1$, and varying fluidity level $\sigma$ (the flexibility parameter).
The results (also shown in Figure \ref{fig:FFLambda} for $\lambda = 1$ and $\sigma = 1$)
 show that the resulting
 performance is marginally better (by only a few percentage points)
 than applying {\em either} non-fluid transformations {\em or} fluid
 transformations.

\noindent{\bf Supporting Different SLA Policies:} IaaS providers outline different SLA policies which reflect their commitment to providing different classes of resource availability defined in terms of ``Uptime Percentage" (UP), which is typically defined as the percentage of time the resource is available during a five minute interval \cite{AmazonS3SLA,AmazonEC2SLA}. These policy types can be represented in our SLA model using  $D$ and $W$ and allow us to apply a wider range of transformation -- compared to SLAs of type $(C,T)$. Thus in our next experiment we measure the effect of supporting different SLA policy types. We use the same model described above to generate non-fluid workloads with $C$ and $T$ and set $\lambda = 1$. We also set $W$ to be equivalent to the number of windows that correspond to a five minute interval, and $D$ to be equal to $\delta$ * W, where $\delta$ defines UP and ranges from 99.9\% to 99.5\%. We also vary the percentage of workloads colocated under a UP policy.

Figure \ref{fig:FigDW} shows the result obtained using FF-NR. As expected, less uptime percentage leads to better CE. In addition, CE increases as we vary the the percentage of worloads colocated under a UP policy. The Figure highlights the use of additional transformation in providing us more flexiblity for colocation, and hence a better performance.

  \begin{figure}
  \begin{center}
  \begin{minipage}{0.42\linewidth}
  	\includegraphics[width=1\textwidth]{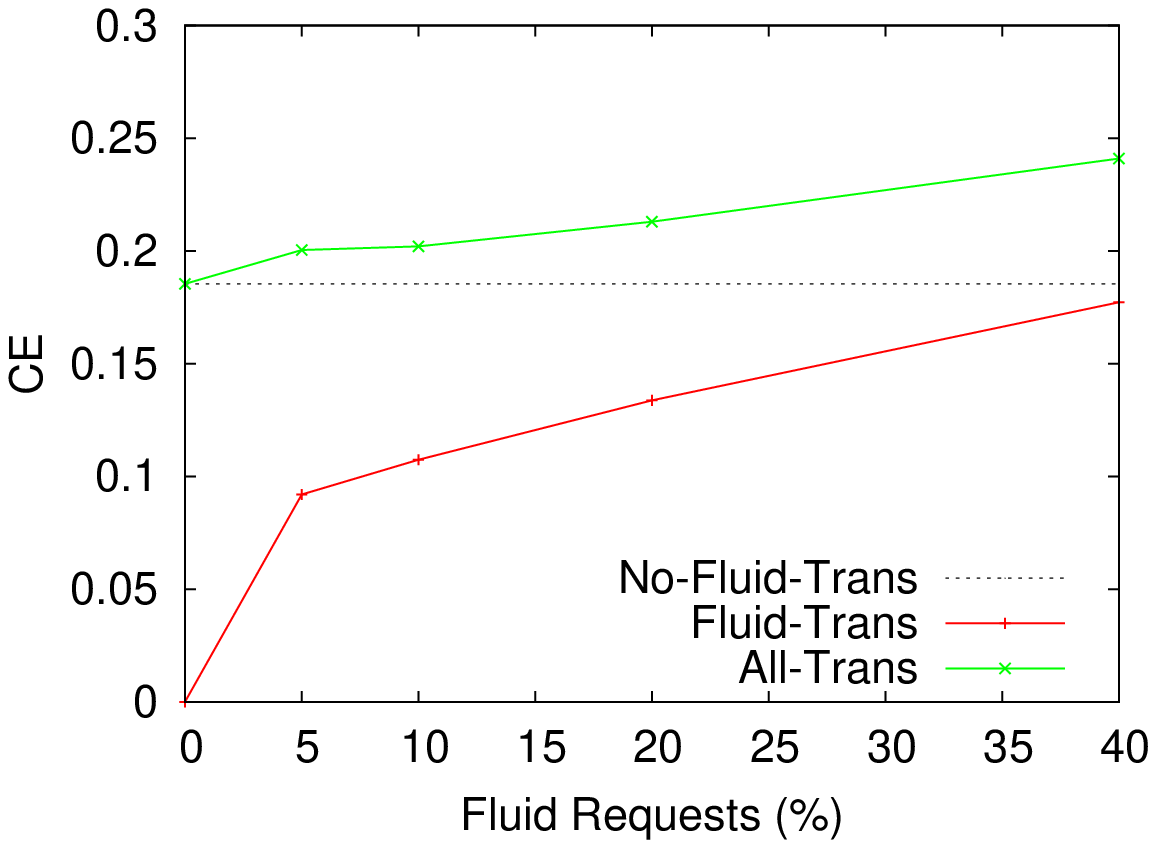}
  	  \caption{Effect of different types of transformations on the CE of
  	  FF-NR ($\lambda = 1$).}\label{fig:FFLambda}
  \end{minipage}
  \begin{minipage}{0.08\linewidth}
   $\quad$
  \end{minipage}
    \begin{minipage}{0.42\linewidth}
    	\includegraphics[width=1\textwidth, right]{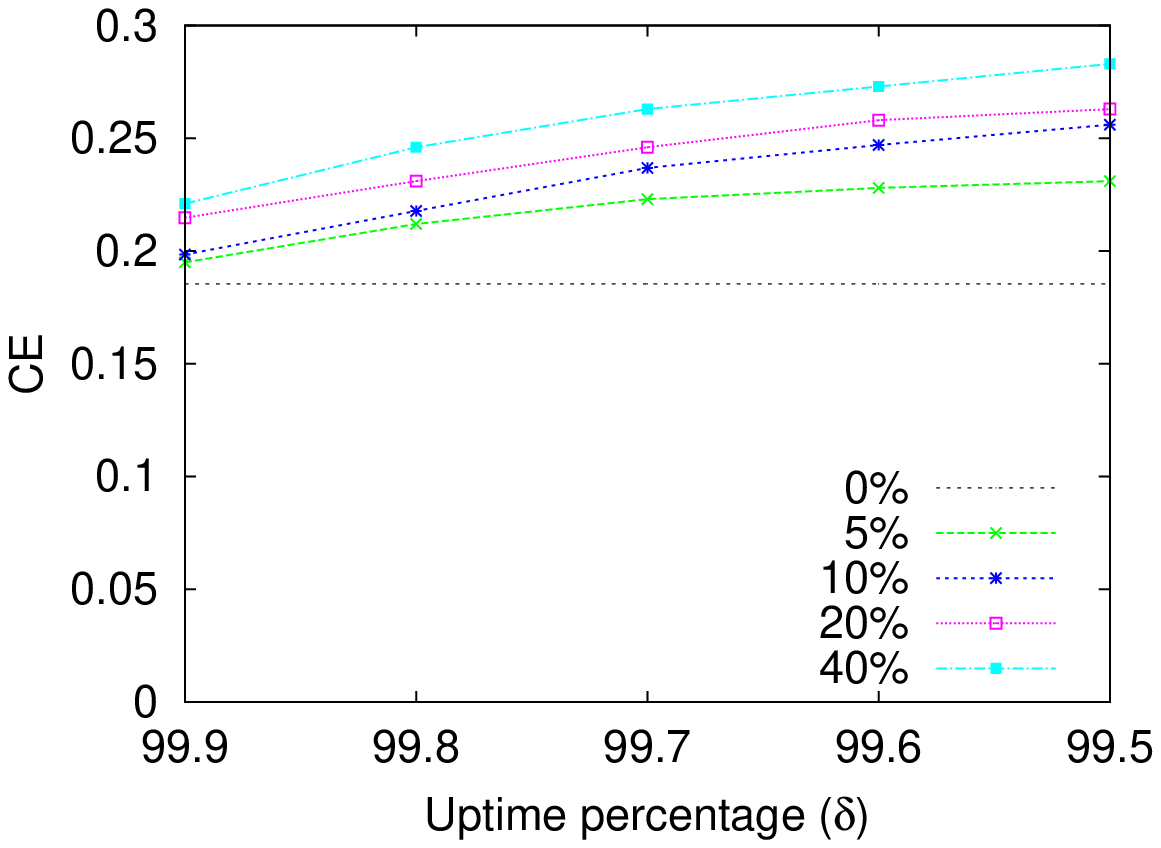}
  		\caption{Effect of additional SLA transformations on the CE of FF-NR ($\lambda = 1$).}\label{fig:FigDW}
    \end{minipage}
      \end{center}
      \vspace{-0.3in}
\end{figure}

\noindent{\bf System Scalability:} In our experiments, the WAS
  component of \morphosys\ was able to handle large clusters of
  resources (disks) -- up to $4,000$. If migration of
  workloads is not enabled, WRS is able to handle even larger clusters
  in an ``online'' fashion. Enabling migration introduces significant
  computational overheads when dealing with large clusters. This makes the
  use of migration in WAS more practical for ``off-line'' (batch)
  use.

In a typical IaaS setting, there might be even more than a few thousand resources
  under management, which would be more than what a single WAS instance can
  handle. We note that in such cases, a practical solution would be to
  group the resources under management into separate clusters, each of
  which is managed by a single WAS.

We note that our measurement of scalability deals only with the
  computational aspect of \morphosys\ (namely, computing efficient
  colocation configuration). In actual deployments, scalability will
  also depend on additional considerations due to system overheads
  that are dependent on the specific setting. For IaaS settings like
  the one considered in the experiments we presented in this paper
  (colocation of streaming servers), one would not expect much
  reconfiguration overheads except that of migrating and merging user stream requests. The effect of such overheads can be reduced using a number of ways including manipulation of stream playout rates 
or by inserting secondary content in the stream \cite{basu1999optimal} (e.g. as done in systems like Crackle \cite{crackle}). However, in other settings involving more significant overheads ({\em e.g.}, the handling of large memory VM images to allow VM migration across hosts), the scalability of \morphosys\ will depend on the efficient management of such  aspects.


\section{Related Work}\label{sec:relatedwork}

\noindent{\bf Service Level Agreements (SLAs):} There has been a
  significant amount of research on various topics related to
  SLAs. The usage of resource management in grids have been considered
  in \cite{czajkowski2002snap}; issues related to	
  specification of SLAs have been considered in
  \cite{Keller2003TheWSLA}; and topics related to the economic aspects of
  SLAs usage for service provisioning through negotiation between
  consumers and providers are considered in \cite{Barmouta2003GridBank}. 
A common characteristic (and/or inherent
  assumption) in the above-referenced body of prior work is that the
  customer's SLAs are immutable. We break that assumption by
  recognizing the fact that there could be multiple, yet functionally
  equivalent ways of expressing and honoring SLAs. Our \morphosys
  $\;$framework utilizes this degree of freedom to achieve
  significantly better colocation.

\noindent{\bf VM Colocation:} VM consolidation and colocation are very
  active research topics  
  that aim to minimize the operating cost
  of data centers in terms of hardware, energy, and cooling, as well
  as providing a potential benefit in terms of achieving higher
  performance at no additional cost. Much work has gone into studying
  the consolidation of workloads across various resources: CPU,
  memory, and network \cite{Cardosa2009Shares,Vatche2010CaaS,  Wood2009Memory,
  Meng2010Improving, podzimek2015analyzing}. Podzimek {\em et al} study the impact of CPU pinning
  on performance interference and energy efficiency \cite{podzimek2015analyzing}.
  Wood {\em et al} \cite{Wood2009Memory} promote
  colocation as a way to minimize memory utilization by sharing
  portions of the physical memory between multiple colocated
  VMs.
  Ishakian {\em et al} \cite{Vatche2010CaaS} presented a colocation
  service which utilizes game theoretic aspects to achieve significant
  cost savings for its (selfish) users. Network-aware consolidations
  have been studied in \cite {Meng2010Improving}. Affinity aware VM Placement has been studied in \cite{chen2016joint}. Colocation has also
  been explored as a means of reducing the power consumption in data
  centers, for example by Cardosa {\em et al} \cite{Cardosa2009Shares}.
  We  note that in all these works, the specification of the resource
  requirements for a VM is static and based on some fixed average
  requested capacities. In our work, the specification of resource
  needs is much more expressive as it allows VMs to control their
  resource allocation time-scale, as well as expose any flexibilities
  VMs may have regarding such timescale.

\noindent{\bf Real-Time Scheduling:} Different scheduling algorithms
  were suggested to deal with scheduling of periodic/aperiodic hard
  real-time and soft-real time tasks \cite{R:Davis:2009d} (and the
  references within). In addition, variants of proportional-share
  scheduling algorithms -- based on the concept of Generalized
  Processor Sharing (GPS) have been suggested \cite{Duda1999Borrowed} -- which allow the integration of different
  classes of applications. These approaches however do not take into
  consideration reservation of resources and fairness in allocating
  resources. The work by Buttazzo {\em et al} \cite{Buttazzo1998Elastic}
  present an elastic task model based on a task defined using a tuple
  $(C, T, T_{min},T_{max},e)$, where $T$ is the period that the task
  requires, $T_{min}$ and $T_{max}$ define the maximum and minimum
  periods that a task can accept. Our SLA model allows us to express
  classes of applications that are more general than the elastic task
  model. Moreover, the SLA transformations that we utilize allow us to
  serve workloads under completely different $(C,T)$ server supplied
  resources.

 \noindent{\bf Hierarchical Scheduling:} hierarchical scheduling (and in particular
  hierarchical CPU scheduling) has been a topic of research for
  over a decade because it allowed multiple scheduling
  mechanisms to co-exist on the same infrastructure -- {\em
  i.e.}, regardless of the underlying system scheduler. For
  example, Goyal {\em et al} \cite{Goyal1996Hierarchical}
  proposed a hierarchical scheduling framework for supporting
  different application classes in a multimedia system;
  Shin and Lee \cite{Shin2005Compositional} further generalized
  this concept, advocating its use in embedded systems. Along
  the same lines, there has been a growing attention to building
  hierarchical real-time scheduling frameworks supporting
  different types of workloads \cite{Regehr2001HLS, Henzinger2006AnInterface}. 

  A common characteristic (and/or inherent assumption) in the
  above-referenced, large body of prior work (which we emphasize
  is not exhaustive) is that the ``clustering'' (or grouping) of
  applications and/or schedulers under a common ancestor in the
  scheduling hierarchy is known {\em a priori} based on domain
  specific knowledge, {\em e.g.}, all applications with the same
  priority are grouped into a single cluster, or all
  applications requiring a particular flavor of scheduling ({\em
  e.g.}, periodic real-time EDF or RMS) are grouped into a
  single cluster managed by the desired scheduling scheme. Given
  such a {\em fixed} hierarchical structure, most of this prior
  body of work is concerned with the schedulability problem --
  namely deciding whether available resources are able to
  support this {\em fixed} structure.


\noindent{\bf Resource Allocation in Distributed Settings:} Different
  approaches have been suggested to deal with resource allocation in
  distributed settings
  \cite{podzimek2017reprint,Netto2007SLABased,foster2002distributed,    
  Buyya2000NimrodG,Czajkowski1999Resource} among many others. In these
  works, the main mechanisms used for providing QoS guarantees to
  users are through resource reservations. Such reservations can be
  immediate, undertaken in advance \cite{foster2002distributed}, or
  flexible \cite{Netto2007SLABased}. To achieve efficient allocation
  and increased resource utilization, these approaches model workloads
  as having a start time and end time. Under such approaches the
  resources allocated to a workload would still be based on a
  percentage reservation, which results in performance variability
  specifically for periodic workload requests. Our work complements
  these models by allowing for an expressive SLA model that admits the
  specification of constraint flexibilities. We believe that providing
  this capability is crucial for the deployment of QoS-constrained
  workloads while at the same time ensuring efficient utilization of
  resources. 

\noindent{\bf Resource Allocation in Cloud Settings:} Efficient scheduling of workloads in the cloud settings is an active topic of research \cite{Ishakian2010AType,jennings2015resource} with central schedulers that focus on soft constraints such as data locality \cite{zaharia2010delay}, deadlines \cite{yang2013bubble,delimitrou2014quasar}, resource guarantees \cite{curino2014reservation} or fairness \cite{ghodsi2011dominant}. Carvalho {\em  et al} \cite{carvalho2014long} analyze historical cloud workload data and motivate for the introduction of a new class of cloud resource offerings.  Curino {\em  et al.} \cite{curino2014reservation} propose a reservation-based system with a declarative language similar to how EC2 resources are requested. In addition to resources requests, our SLA language incorporates predictable (timely) access to resources and allows for the customers to provide their flexibility. Yang {\em  et al} \cite{yang2013bubble} developed an online scheme that detects memory pressure and finds colocations that avoid interference on latency-sensitive applications. Zaharia {\em  et al} \cite{zaharia2010delay} use delayed scheduling of tasks to capitalize on data locality. Results show decreased job turnaround times. Chen {\em  et al} \cite{chen2014distributed} proposed a long term load balancing VM migration algorithm based on finite-markov decision process with the goal of reducing SLA violations. The \morphosys\ framework uses safe transformations of workloads as a tool that complements these methods  to enable  efficient colocation while ensuring that the scheduled workloads have predictable access to resources.

\section{Conclusion}\label{sec:conclusion}
The value proposition of virtualization technologies is highly dependent on our ability to identify judicious mappings of physical resources to virtualized instances that could be acquired and consumed by applications subject to desirable performance ({\em e.g.}, QoS) bounds. These bounds are often spelled out as a Service Level Agreement (SLA) contract between the resource provider (hosting infrastructure) and the resource consumer (application workload). By necessity, since infrastructure providers must cater to very many types of applications, SLAs are typically expressed as fixed fractions of resource capacities that must be allocated (or are promised) for {\em unencumbered} use. That said, the mapping   between ``desirable performance bounds'' and SLAs is not unique. Indeed, it is often the case that {\em multiple} SLA expressions might be functionally equivalent with respect to the satisfaction of these performance bounds. Having the flexibility to transform SLAs from one form to another in a manner that is {\em safe} would enable hosting solutions to achieve significant economies of scale.

In this paper, we proposed a new SLA model for managing  QoS-constrained workloads in IaaS settings. Our SLA model supports an expressive specification of the requirements for various classes of applications, thus facilitating auditability and performance predictability using simple measurement techniques. We presented the architectural and algorithmic blueprints of a framework for the deployment of dynamic colocation services. \morphosys\  utilizes workload SLA transformations (exploiting any flexibility therein) for efficient management of QoS-constrained workloads in the cloud. We evaluated our framework by considering a cloud storage service scenario, and performed extensive evaluation using real video traces. The results reported in this paper -- which suggest significant reduction in unallocated (wasted) resources of up to 60 percent -- underscore the potential from deploying \morphosys-based services.


\bibliographystyle{elsarticle-num}
\bibliography{bibtex}

\appendix
\section{}

\begin{lem} \label{k1overlap}
Given the periods $T$ and $T'$ such that $T \leq T'/2$.
Then an interval of length $T'$ would contain at least $(K-1)$ intervals
of length $T$ where $K = \lfloor T'/T \rfloor$.
\end{lem}
\begin{proof}
Figure \ref{k-1case} highlights the existence of a schedule such that $T'$ overlaps $(K-1) * T$ intervals where $T'=7$ and $T=3$ and $K = 2$.
Assume the existance of a schedule where $T$ overlaps only with $(K-2) * T'$ intervals as shown in Figure \ref{k-2case}.

  \begin{figure*}[htp]
  \begin{minipage}{0.42\linewidth}
  	\includegraphics[width=1\textwidth, left]{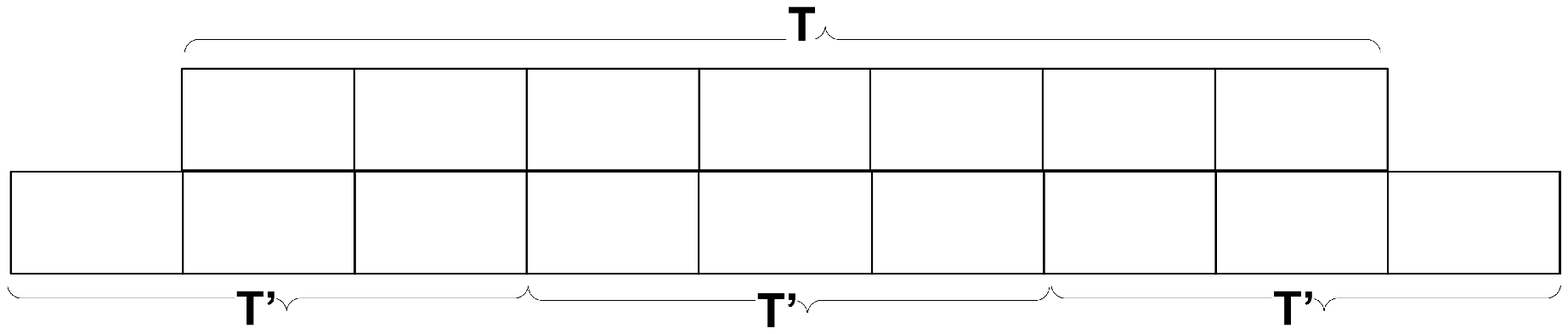}\\
  	\caption{Example of $T'$ overlapping $(K-1) * T$ intervals} \label{k-1case}
  \end{minipage}
  \begin{minipage}{0.10\linewidth}
   $\quad$
  \end{minipage}
    \begin{minipage}{0.42\linewidth}
    	\includegraphics[width=1\textwidth, right]{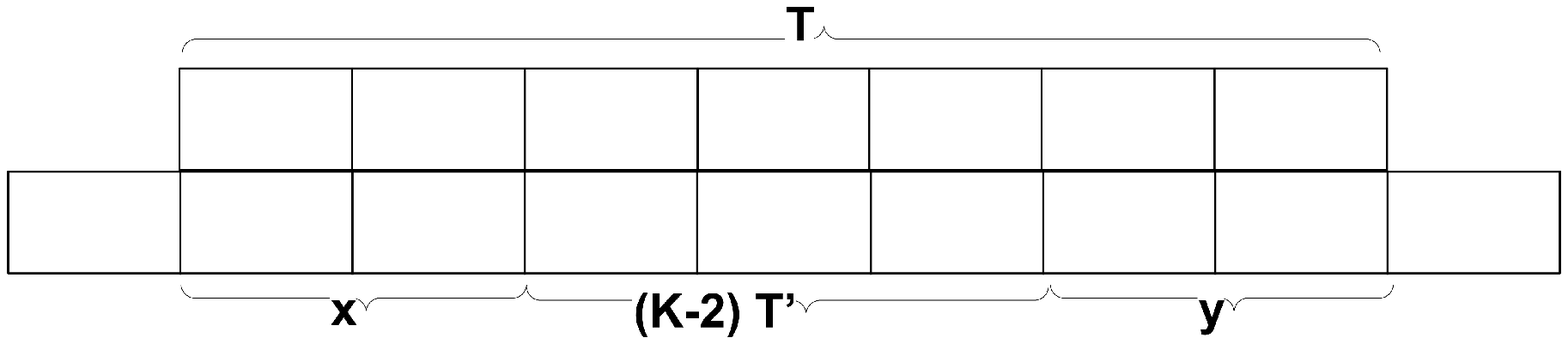}\\
    	\caption{\label{k-2case}Schedule where $T$ overlaps with $(K-2) * T'$}
    \end{minipage}
\end{figure*}

We observe that $T = x + (K-2) * T' + y$. From the definition we have $K * T' \leq T$. Therefore, $K * T' \leq x + (K-2) * T' + y \leq x + y.$
But by definition $x < T'$ and $y <T'$ -- a contradiction.
\end{proof}

\begin{lem} \label{Necessary}
Given $(C,T,D,W) \lhd (C,T,D',W')$  or $(C,T,D,W) \lhd (C,T,D',W')$, it is necessary for $D/W \leq D'/W'$.
\end{lem}
\begin{proof}
We provide counter-examples for all possible values of $D$ and $W$.
\begin{itemize}
  \item $D \leq D'$ and $W \leq W'$. Then unless $D/W \leq
      D'/W'$, we could have $(C,T,D,D) \lhd
      (C,T,D',W')$ or $(C,T,D,D) \lhd (C,T,D',W')$.
      Contradiction.
  \item $D \geq D'$ and $W \leq W'$ Then unless $D/W \leq
      D'/W'$, we could have $(C,T,W,W) \lhd
      (C,T,D',W')$ or $(C,T,W,W) \lhd (C,T,D',W')$.
      where $D = W$. Contradiction.
  \item $D \geq D'$ and $W \geq W'$ Then unless $D/W \leq
      D'/W'$, we could have $(C,T,W,W) \lhd
      (C,T,D',W')$ or $(C,T,W,W) \lhd (C,T,D',W')$.
      Contradiction.
\end{itemize}
   Thus $D/W \leq D'/W'$.
\end{proof}

\end{document}